\documentclass[11pt,reqno]{amsart}

\usepackage{amsthm, amsmath, amssymb, amsfonts}

\usepackage{graphicx, epsfig}

\usepackage{bbm}

\usepackage{epstopdf}

\usepackage{float}
\usepackage{enumerate}

\usepackage[colorlinks=true, pdfstartview=FitV, linkcolor=blue,
            citecolor=blue, urlcolor=blue]{hyperref}
\usepackage[usenames]{color} 

\usepackage[round]{natbib}

\usepackage{url}
\makeatletter\def\url@leostyle{%
 \@ifundefined{selectfont}{\def\UrlFont{\sf}}{\def\UrlFont{\scriptsize\ttfamily}}} \makeatother
 
\definecolor{Red}{rgb}{0.8,0,0.1}

\usepackage{multirow}
\usepackage{listings}
\newsavebox{\mybox}

\usepackage{xcolor}
\usepackage{booktabs,nicefrac}
\usepackage{array}

\definecolor{RED}{RGB}{255,0,0}

\usepackage[font={footnotesize}]{caption}

\newcommand{\1}{\mathbbm{1}}            % preferable way of writing indicator function
\newcommand{\abs}[1]{\left\vert#1\right\vert}   % absolute value
\newcommand{\wsum}{\mathcal{E}}			%directed sum symbol

\DeclareMathOperator{\Cov}{Cov}          % \Cov for covariance
\def\bN{\mathbb{N}}
\def\bE{\mathbb{E}}
\def\bR{\mathbb{R}}

\newtheorem{theorem}{Theorem}

\newtheorem{lemma}[theorem]{Lemma}
\newtheorem{corollary}[theorem]{Corollary}

\theoremstyle{definition}

\theoremstyle{remark}
\newtheorem{remark}[theorem]{Remark}

\numberwithin{equation}{section}
\numberwithin{theorem}{section}

\author{Damian Jelito \and Marcin Pitera}
\address{Institute of Mathematics, Jagiellonian University, \L{}ojasiewicza 6, 30-348 Cracow, Poland}
\email{\url{damian.jelito@im.uj.edu.pl}}
\email{\url{marcin.pitera@im.uj.edu.pl}}

\begin{document}

\title[New fat-tail normality test based]{New fat-tail normality test based on conditional second moments with applications to finance}

%%%%%%%%%%%%%%%%%%%%%%%%%%%%%%%%%%%%%%%%%%%%%%%%%%%%%%%%%%%%%%%%%%%%%%%%%%%%%%%%%%%%%%%%%%%%%%%%%

%\twocolumn[
%  \begin{@twocolumnfalse}
%%\vspace{-0.2cm}
%    \maketitle
%\vspace{-0.7cm}
\begin{abstract} 
In this paper we introduce an efficient fat-tail measurement framework that is based on the conditional second moments.  We construct a goodness-of-fit statistic that has a direct interpretation and can be used to assess the impact of fat-tails on central data conditional dispersion. Next, we show how to use this framework to construct a powerful normality test. In particular, we compare our methodology to various popular normality tests, including the Jarque--Bera test that is based on third and fourth moments, and show that in many cases our framework outperforms all others, both on simulated and market stock data. Finally, we derive asymptotic distributions for conditional mean and variance estimators, and use this to show asymptotic normality of the proposed test statistic. 
\vspace{0.2cm}

\end{abstract}
\vspace{-5cm}
%  \end{@twocolumnfalse}
%]
\maketitle
%%%%%%%%%%%%%%%%%%%%%%%%%%%%%%%%%%%%%%%%%%%%%%%%%%%%%%%%%%%%%%%%%%%%%%%%%%%%%%%%%%%%%%%%%%%%%%%%%

\section{Introduction}

It has been recently shown in \cite{JawPit2015} that for a normal random variable and a unique ratio close to 20/60/20 the conditional dispersion in the tail sets is the same as in the central set. In other words, if we split big normal sample into three sets -- one corresponding to worst 20\% outcomes, one corresponding to the middle 60\% outcomes, and one corresponding to best 20\% outcomes -- then the conditional variance on those subsets is approximately the same.

In this paper we show that this property could be used to construct an efficient goodness-of-fit testing framework that has a direct (financial) interpretation. The impact of tail dispersion on central dispersion is a natural measure of tail heaviness and can serve as an alternative to other methods which are typically based on tail limit analysis or higher order moments; see \cite{Car2009b} and \cite{JarBer1980}. In particular, in contrast to the Jarque--Bera normality test that is based on third and fourth moments, our test relies on the conditional second moments which are often easier to estimate.

Testing for normality has a long history and many remarkable methods have been developed. This includes general distribution-fit frameworks like Anderson--Darling test based on the distance between theoretical and empirical distribution function \citep{AndersonDarling1954} or Shapiro--Wilk test relying on the regression coefficient \citep{ShapiroWilk1965}; see \cite{Mad2012}, \cite{Henze2002} or \cite{Thode2002} for a comprehensive overview of normality testing procedures. 

Most empirical studies suggest that the normality tests should be chosen carefully as their statistical power varies depending on the context; see e.g. \cite{ThadBun2007}, \cite{RomDelCost2010} or \cite{BroDavis2016}. This is why the existing procedures are constantly refined and new ones are being developed; for example, a recent revision of the Jarque--Bera testing framework based on the second-order analogue of skewness and kurtosis could be found in \cite{DesLafMich2018}.

The approach presented in this paper draws attention to interesting and previously not exploited aspect of normal distributions that could be used for efficient normality testing. In particular, we show that our approach outperforms and/or complements multiple benchmark normality testing frameworks when popular (financial) alternative distributions, such as Student's t or logistic, are considered. More explicitly, we show that our test usually has the best power if one expects that a sample comes from a symmetric distribution which have heavier (or lighter) tails in comparison with the normal distribution. We illustrate that with a financial market data example; see Section~\ref{S:2} for details.

Finally, it is worth mentioning that the 20/60/20 division leads to a very accurate data clustering when it comes to the tail assessment performed in reference to the central set normality assumption. Consequently, our method could be embedded into data analytics frameworks based on cluster analysis, help to refine data mining techniques, etc; see \cite{Rom2003,KauRou2009,HaiBlaBabAnd2013} for an overview.  In fact, the good performance of our test statistic on market data could be linked to a popular financial \textit{stylised fact} saying that typical financial asset returns can be seen as normal, but the extreme returns are more frequent and with greater magnitude than the ones resulting from the normal fit; see \cite{Cont2001} and \cite{SheQia2010} for details. 

This paper is organised as follows: In Section~\ref{S:206020_Rule} we briefly recall the concept of the 20-60-20 Rule, while in Section~\ref{S:stat.test} we outline the construction of the test statistic and discuss its basic properties. Section~\ref{S:power} provides a high-level discussion about the test power, and Section~\ref{S:asymptotic} discusses in details mathematical background including derivation of the asymptotic distribution of the proposed test statistic. Next, in Section~\ref{S:2}, we present a simple market data case-study and discuss application of our framework in the financial context. We conclude in Section~\ref{S:conculde}. For brevity, we moved the closed-form formula for the normalising constant introduced in Section~\ref{S:stat.test} to Appendix~\ref{A:rho}.

%%%%%%%%%%%%%%%%%%%%%%%%%%%%%%%%%%%%%%%%%%%%%%%%%%%%%%%%%%%%%%%%%%%%%%%%%%%%%%%%%%%%%%%%%%%%%%%%%

\section{The 20-60-20 Rule for the univariate normal distribution}\label{S:206020_Rule}
Let us assume that $X$ is a normally distributed random variable. We define left, middle, and right partitioning sets of $X$ by
\begin{equation}\label{eq:LRM} 
L:= \left(-\infty, F_X^{-1}(0.2)\right],\, M:= \left(F_X^{-1}(0.2), F_X^{-1}(0.8)\right),\,R:= \left[F_X^{-1}(0.8),+\infty\right),
\end{equation}
where $F_X^{-1}(\alpha)$ is the $\alpha$-quantile of $X$. It has been shown in~\cite{JawPit2015} that
\begin{equation}\label{rq:LMR}
\sigma^2_L =\sigma^2_M=\sigma^2_R,
\end{equation}
for this unique 20/60/20 ratio, where $\sigma^2_A$ denotes the conditional variance of $X$ on set $A$\footnote{In fact, this equality is true for the ratio very close to 20/60/20, i.e. for upper and lower quantiles equal to approximately $0.198$. For transparency, we have decided to use the rounded numbers here; see Section~\ref{S:asymptotic} for details.}.

This specific division together with the associated set of equalities in \eqref{rq:LMR} create a dispersion balance for the conditioned populations. This property might be linked to the statistical phenomenon known as the {\it 20-60-20 Rule}: a principle that is widely recognised by the practitioners and used e.g. for efficient management or clustering. In fact, a similar statement is true in the multivariate case: the conditional covariance matrices of multivariate normal vector are equal to each other, when the conditioning is based on the values of any linear combination of the margins, and 20/60/20 ratio is maintained. For more details, see \cite{JawPit2015} and references therein.

%%%%%%%%%%%%%%%%%%%%%%%%%%%%%%%%%%%%%%%%%%%%%%%%%%%%%%%%%%%%%%%%%%%%%%%%%%%%%%%%%%%%%%%%%%%%%%%%%

\section{Test statistic}\label{S:stat.test}
Let us assume we have a sample from $X$ at hand. Then, based on~\eqref{rq:LMR}, we define a test statistic
\begin{equation}\label{eq:N}
N:= \frac{1}{\rho}\left(\frac{\hat\sigma^2_L-\hat\sigma^2_M}{\hat\sigma^2} +\frac{\hat\sigma^2_R-\hat\sigma^2_M}{\hat\sigma^2}\right)\sqrt{n}\,,
\end{equation}
where $\hat\sigma^2$ is the sample variance, $\hat\sigma^2_A$ is the conditional sample variance on set $A$ (where the conditioning is based on empirical quantiles), $n$ is the sample size, and $\rho\approx 1.8$ is a fixed normalising constant; see Figure~\ref{fig:Rcode} for the {\bf R} implementation code. We refer to Section~\ref{S:asymptotic} for more details including rigorous definitions of conditional variance $\hat\sigma^2_A$, constant $\rho$, etc.

\begin{lrbox}{\mybox}
\begin{lstlisting}
Test.N <- function(x){
	  n  <- length(x)
	  q1 <- quantile(x,0.2)
  	q2 <- quantile(x,0.8)
  	low  <- x[x <= q1]
 	 med  <- x[x > q1 & x < q2]
  	high <- x[x >= q2]
  	N <- var(low)+var(high)-2*var(med)
  	N <- N * sqrt(n)/(var(x)*1.8)
 return(N)}
\end{lstlisting}
\end{lrbox}

\begin{figure}[h]
\hspace{-4cm}
\scalebox{0.7}{\usebox{\mybox}}
\caption{Simplified {\bf R} source code that was used to create a function that computes the test statistic $N$ given input sample $x$}\label{fig:Rcode}
\end{figure}

It is not hard to see that under the normality assumption $N$ is a pivotal quantity. Furthermore, in Section~\ref{S:asymptotic} we show that the distribution of $N$ is asymptotically normal; see Theorem \ref{th:asymptotic_distr_test_stat} therein. In Figure~\ref{F:fig2}, we illustrate this by computing the Monte Carlo density of $N$ under the normality assumption for samples of size $50$, $100$, and $250$.

\begin{figure}[htp!]
\begin{center}
\begin{minipage}{.40\linewidth}\centering
\includegraphics[width=0.95\textwidth]{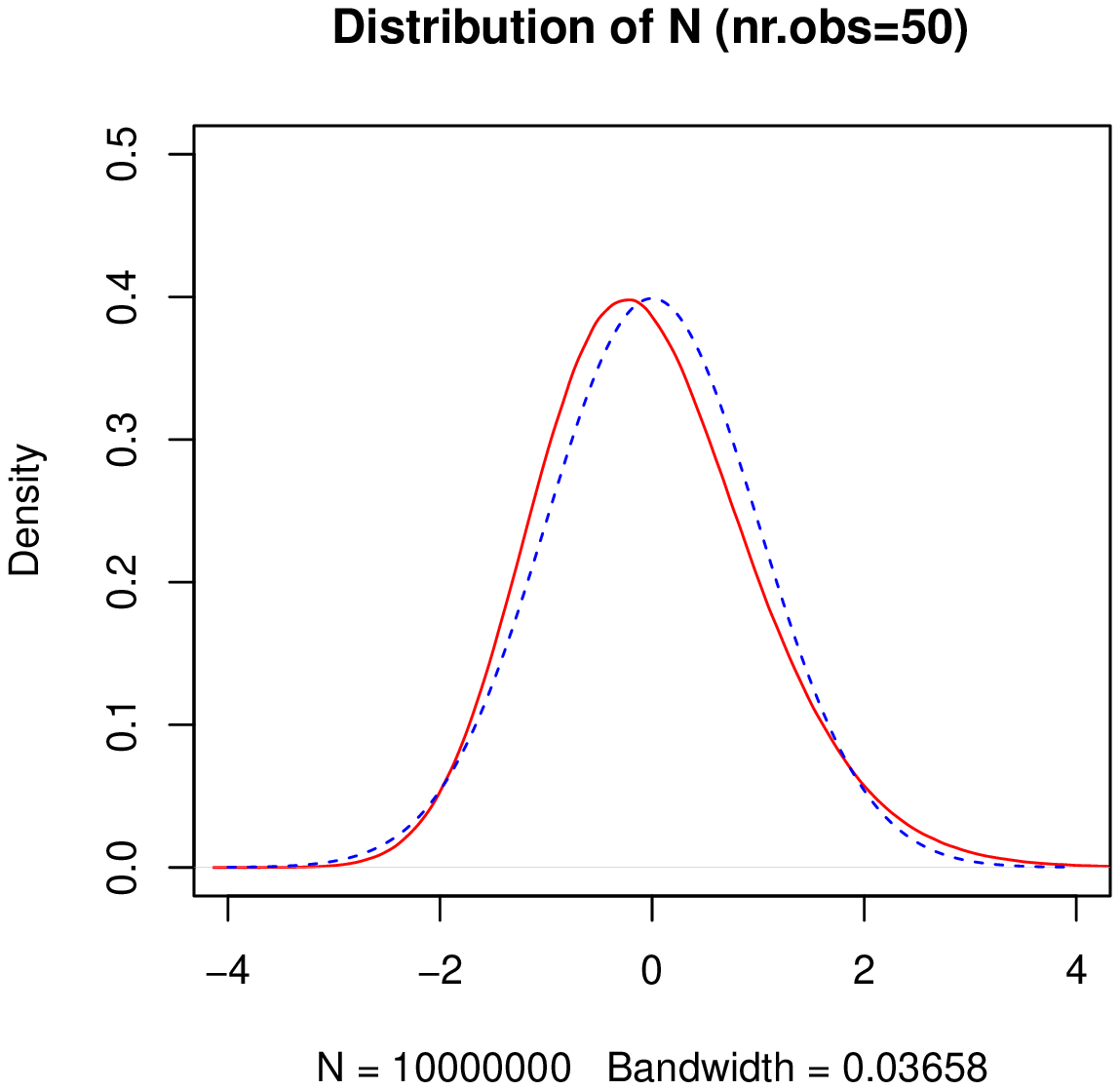}
\end{minipage}
\begin{minipage}{.40\linewidth}\centering
\includegraphics[width=0.95\textwidth]{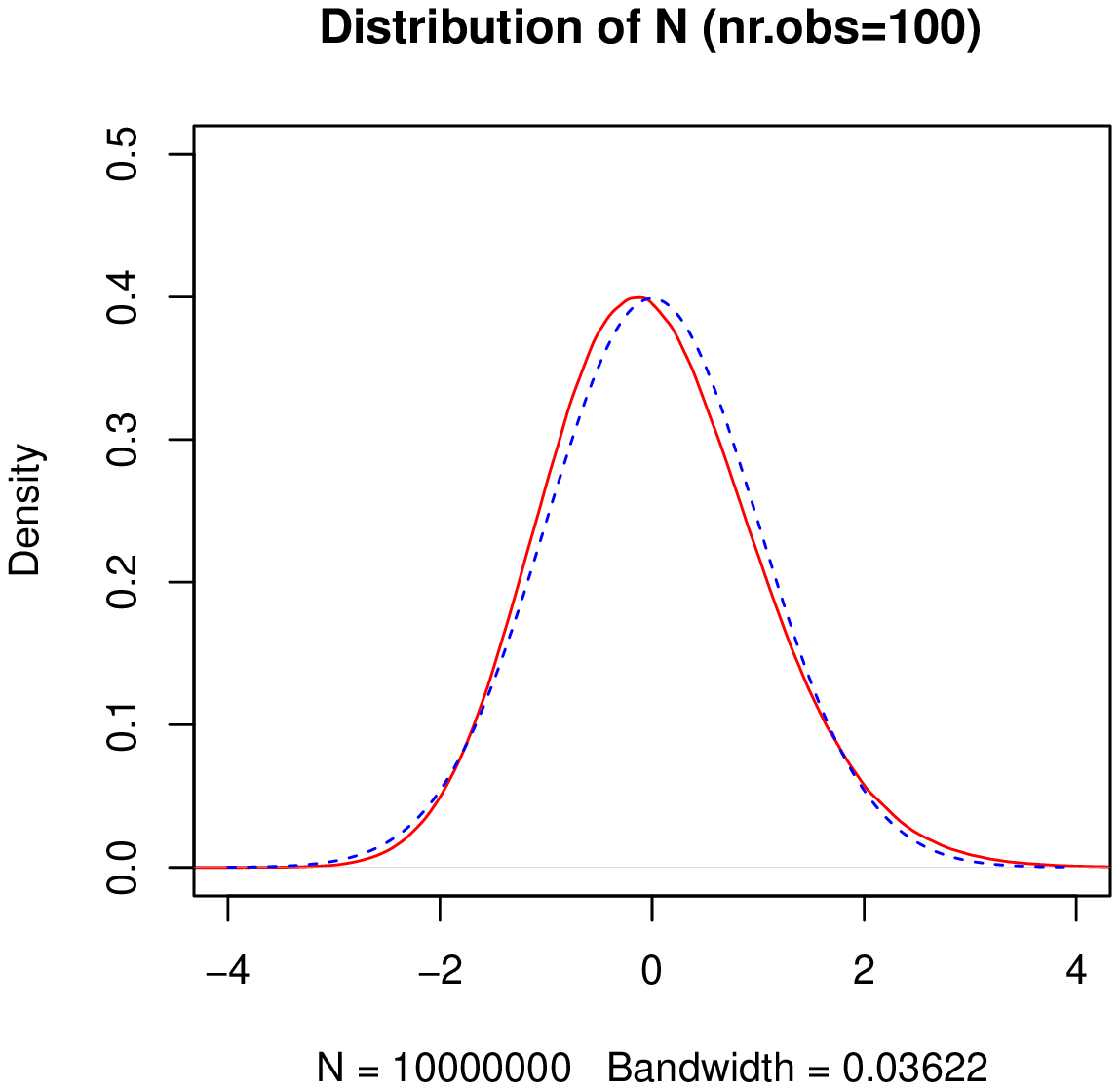}
\end{minipage}\\
\begin{minipage}{.40\linewidth}\centering
\includegraphics[width=0.95\textwidth]{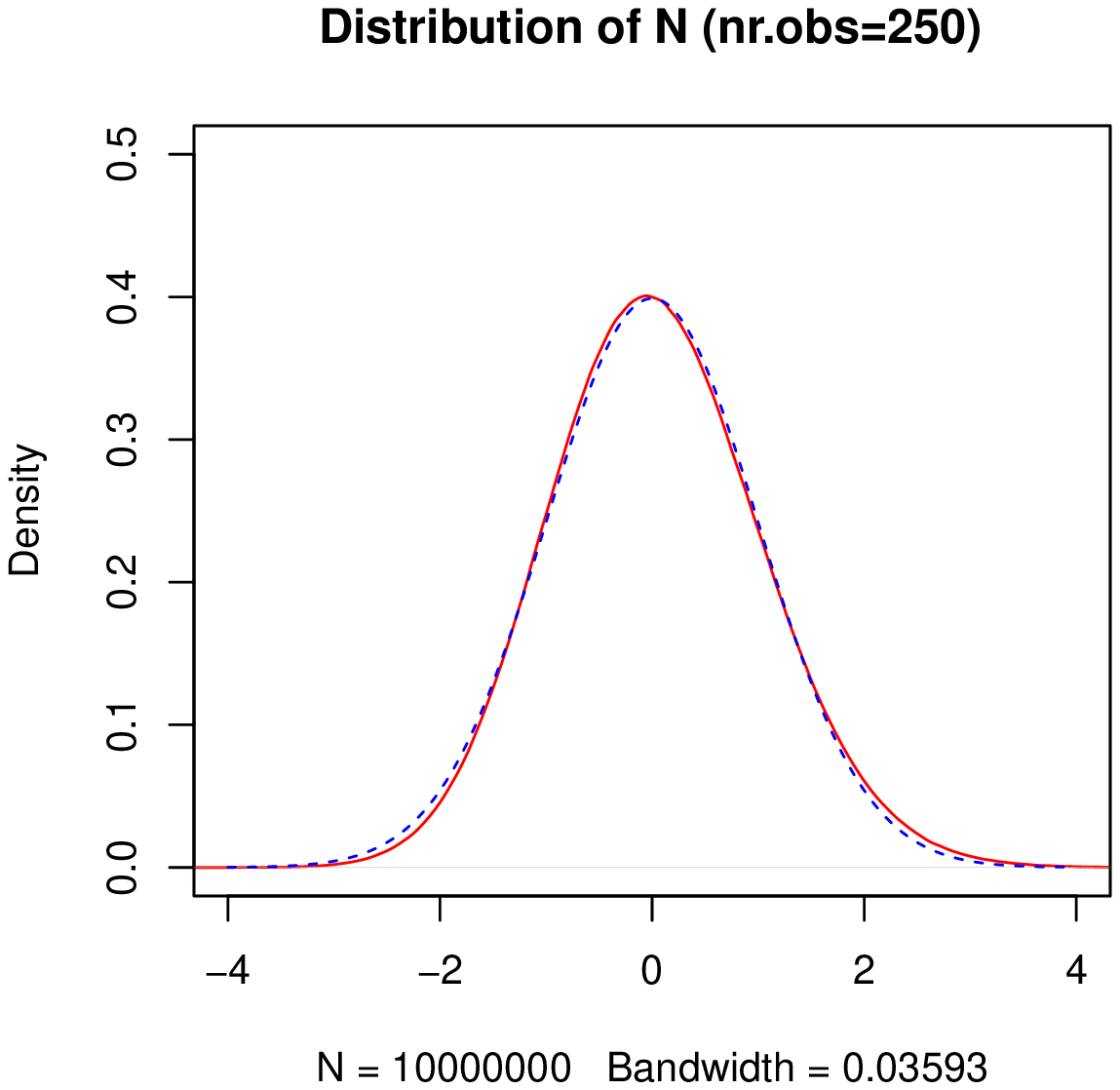}
\end{minipage}
\begin{minipage}{.40\linewidth}\centering
\end{minipage}
\begin{minipage}{.40\linewidth}\centering
\scalebox{0.70}{
\begin{tabular}{c|c|c|c}\toprule
$\alpha$ &  $n$ & $\Phi^{-1}(1-\alpha)$  &  $F^{-1}_{n}(1-\alpha)$ \\ \midrule
               &   50  &   		       		   & 2.68 \\
1.0\%      & 100  & 2.33    		           &  2.57\\
               & 250  &    		                   &  2.51\\\midrule
               &   50  &   		       		   & 2.18 \\
2.5\%      & 100  & 1.96    		           &  2.12\\
               & 250  &    		                   &  2.09\\\midrule
               &   50  &   		       		   & 1.77 \\
5.0\%      & 100  & 1.64    		           &  1.74\\
               & 250  &    		                   &  1.74\\\bottomrule
\end{tabular}
}
\end{minipage}
\end{center}
\caption{The distribution of $N$ under the normality assumption for $n=50,100,250$, for strong Monte Carlo sample of size 10\,000\,000. The obtained empirical density (solid curve) is very close to standard normal density (dashed curve); the table compares selected empirical quantiles with theoretical normal quantiles.}
\label{F:fig2}
\end{figure}

Test statistic $N$ has a clear interpretation: the difference between tail and central conditional variances could be seen as a measure of tail fatness, i.e. the bigger the value of $N$, the fatter the tails.

To illustrate this we compute the values of $N$ for bigger sample size, $n=500$, for three different fat-tail distributions and three different slim-tail distributions.
For fat-tail comparison we picked logistic, Student's t with five degrees of freedom, and Laplace distributions, while for slim-tail comparison we considered generalised normal distribution with shape parameter $s\in \{2.5, 3, 5\}$; the (standardised) generalised normal density for $s\in\bR_{+}$ is given by
\begin{equation}\label{eq:GN_density}
f(x|s):=\frac{s}{2\Gamma(1/s)}\exp\left(-|x|^{s}\right),\quad x\in\bR;
\end{equation}
we refer to \cite{Nad2005} or \cite{TumKeaBal2016} for more details. The results presented in Figure~\ref{F:candle} confirm that the behaviour of $N$ is as expected.

Based on statistic $N$ values, one can construct a one-sided or two-sided statistical test with normality ($N=0$) as a null hypothesis. For brevity, we refer to such test as {\it N normality test} or simply {\it N test}.

\begin{figure}[htp!]
\begin{center}
\includegraphics[width=0.48\textwidth]{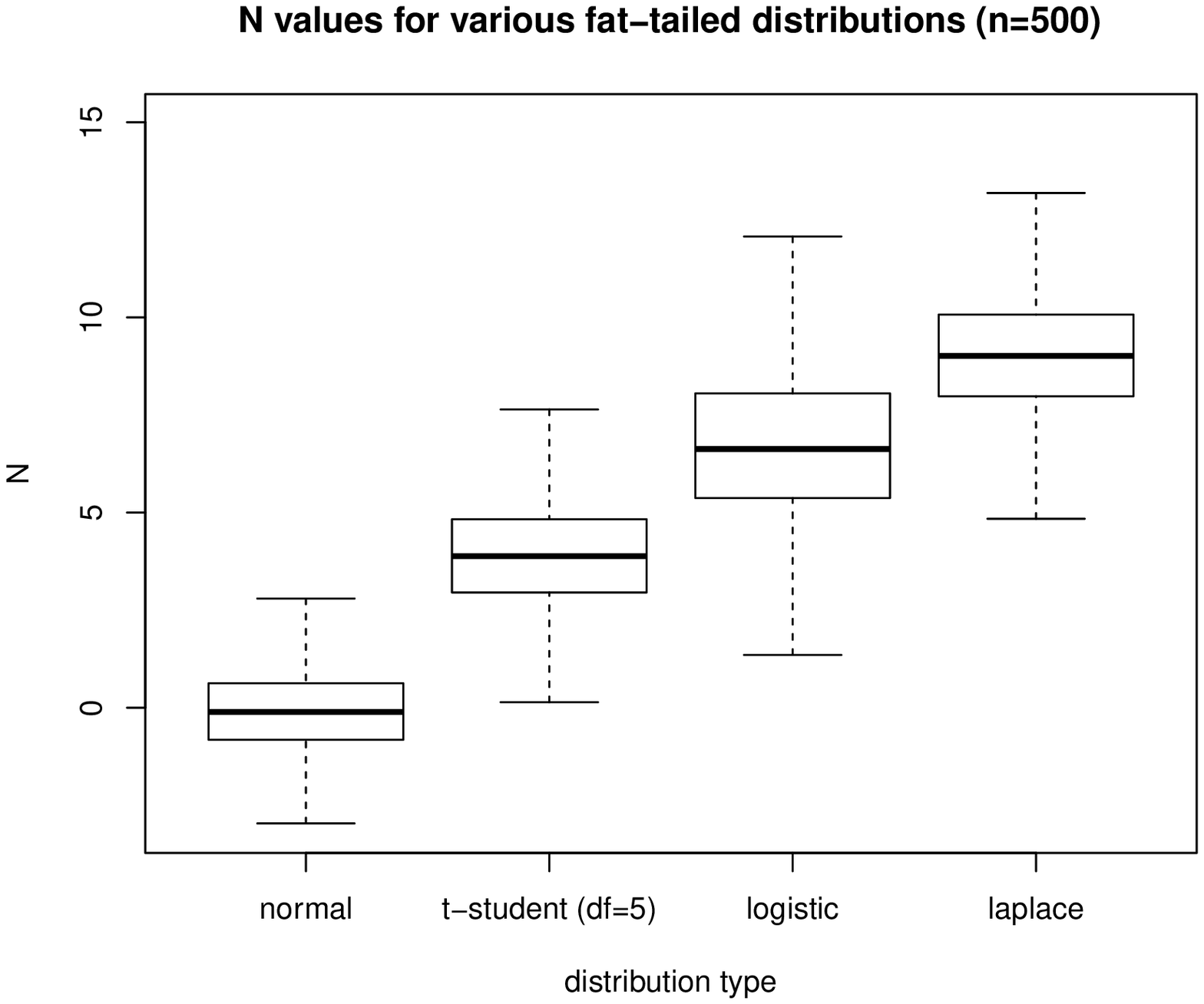}
\includegraphics[width=0.48\textwidth]{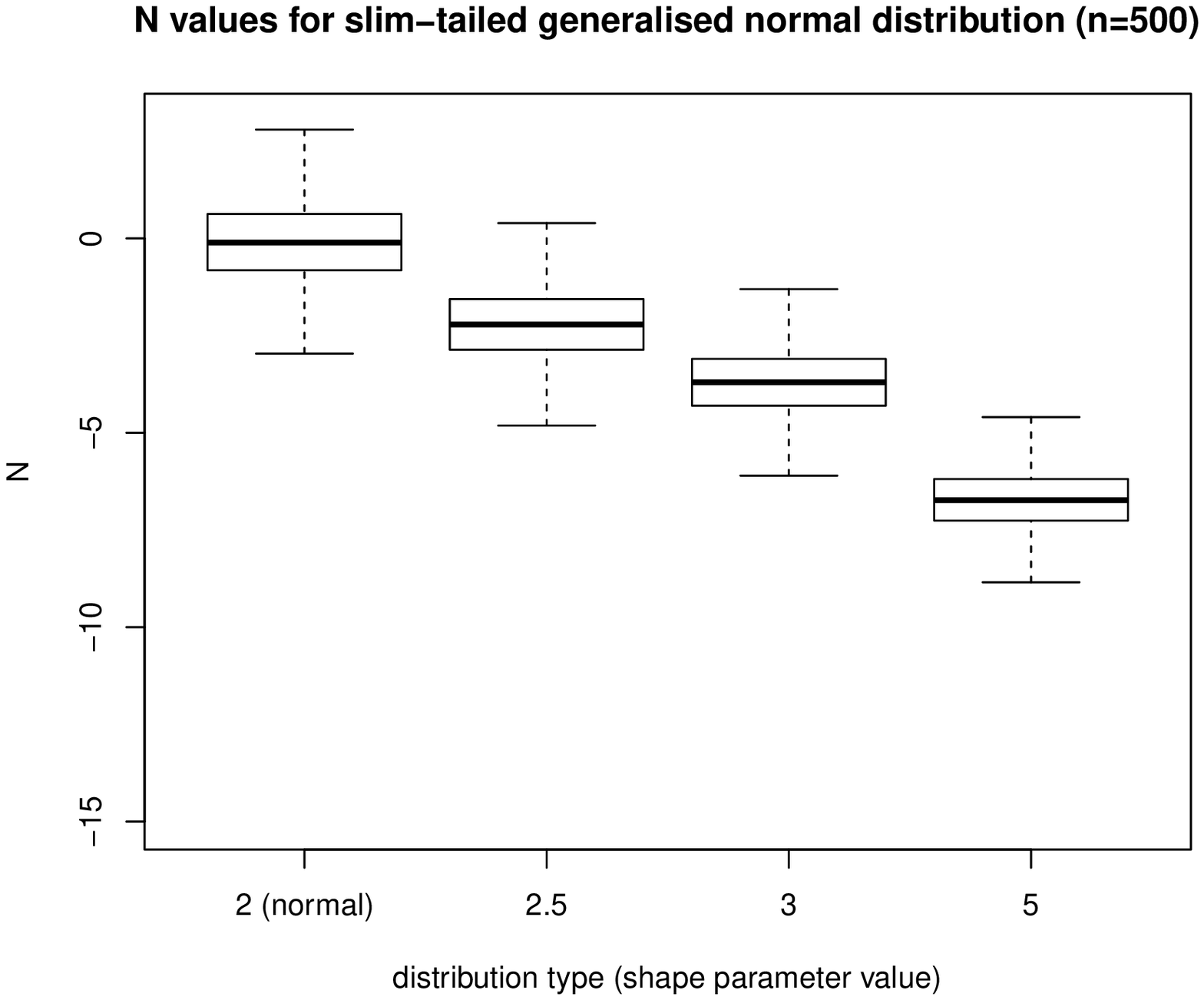}
\end{center}
\caption{Boxplots of $N$ for samples from different distributions. The results for three fat-tailed distributions (Logistic, t-Student with $v=5$, and Laplace) are presented in the left plot. The results for three slim-tailed distributions (generalised normal with $s\in\{2.5,3,5\}$) are presented in the right plot. For transparency, we rescaled all distributions to have zero mean and unit variance, and added the results for normal distribution. The boxplots are based on strong Monte Carlo sample of size 10\,000, where each simulation is of size $n=500$.}
\label{F:candle}
\vspace{-0.5cm}
\end{figure}

%Consequently, our framework is a nice alternative to the standard measures of fat-tails that are based on kurtosis. In a nutshell, instead of estimating the fourth moment of the whole distribution we compare conditional second moments. This makes $N$ very attractive from the practical point of view, as the measurement of fat-tails in reference to the normal framework is the key task in many risk measurement problems. 

%%%%%%%%%%%%%%%%%%%%%%%%%%%%%%%%%%%%%%%%%%%%%%%%%%%%%%%%%%%%%%%%%%%%%%%%%%%%%%%%%%%%%%%%%%%%%%%%%
%\clearpage
\section{Power of the test}\label{S:power}

In this section we check the power of the proposed $N$ test in a controlled environment. We focus on symmetric distributional alternatives (used e.g. in finance) when one wants to abandon the normality assumption due to fat-tail or slim-tail phenomena. Namely, we consider the Cauchy distribution, the Logistic distribution, the Laplace distribution, the Student's t distribution with $v\in\{2, 5, 10, 20, 30\}$ degrees of freedom parameter, and the generalised normal distribution (GN) with the shape parameter $s\in \{1.5,2.5,3,5,10\}$. Note that GN with $s=1$ and $s=2$ correspond to Laplace and normal distribution, respectively; see \eqref{eq:GN_density} for the GN density definition. In all cases, the location parameter is set to $0$ and the scale parameter is set to $1$.

For completeness, we compare test $N$ results with well-established benchmark normality tests: Jarque--Bera test, Anderson--Darling test, and Shapiro--Wilk test. It should be noted that in contrast to these frameworks, statistic $N$ allows one to consider a specific heavy-tail (or slim-tail) alternative, i.e. positive (or negative) values of $N$ point out to heavy-tail (or slim-tail) alternative hypothesis. Consequently, we decided to construct right-sided (or left-sided) critical region for the fat-tailed (or slim-tailed) distributions. Nevertheless, for completeness, we include the results for two-sided critical regions in all cases.

For all alternative distribution choices we consider four different sample sizes, i.e. $n=20,50,100,250$. For each $n$, we simulate 2\,000\,000 strong Monte Carlo sample, and check for what proportion of simulations the tests reject normality at significance level $\alpha=5\%$.

All computations are performed in {\bf R 3.5.2}. For benchmark normality testing we use multiple add-on {\bf R} packages including \textit{gnorm} (for GN simulation), \textit{stats} (for Shapiro--Wilk test), \textit{nortest} (Anderson--Darling test), and \textit{tseries} (Jarque--Bera test). For better comparability, we used test-wise simulated rejection thresholds instead of theoretical $p$-values returned by {\bf R} functions; for computations, we used big strong Monte Carlo sample of size 10\,000\,000. In particular, note that while Jarque--Bera test has asymptotic $\chi^2$ distribution (with 2 degrees of freedom) under normality, this approximation may be inaccurate for small samples, and lead to non-meaningful (non-adjusted) p-values. 

It should be noted that our framework specification is consistent with the one presented in \cite{DesLafMich2018}, where a comprehensive normality tests comparison is made. In particular, results presented in Appendix C therein are perfectly consistent with results presented here (for all benchmark tests).\footnote{Please note that results presented for  Asymmetric Power Distribution (APD) with symmetry parameter $\alpha=0.5$ correspond to results presented for Generalised Normal (GN) distribution; we refer to Section 2 in \cite{DesLafMich2018} for details.}

For transparency, we have decided to consider fat-tail and slim-tail case separately. The results for fat-tailed distributions are presented in Table~\ref{T:tab3}. The best performance of the right-sided test $N$ could be observed for almost all considered distributions. The fatter the tails the bigger the absolute difference between $N$ test power and JB test power, which could be considered as a second best choice. To check whether test statistic $N$ brings some novel results, we decided to check the proportion of simulations on which the normality assumption was rejected uiquely by $N$ among all considered tests; for comparison purposes, we checked the same for all other tests. The results for three selected distributions are presented in Table~\ref{T:tab3b}. It can be observed, that the unique rejection proportion for test $N$ is the highest among all tests and point out to the fact that statistic $N$ is taking into account sample properties that are not exploited by other tests.\footnote{The results for other distributions are consistent with those presented in Table~\ref{T:tab3b} and could be obtained from authors upon request.} Finally, it should be noted that the performance of two-sided $N$ test is also quite good: the test outperforms both AD and SW in most of the considered cases.

The results for slim-tailed distributions are presented in Table~\ref{T:tab3c}. Both left-sided and two-sided test based on $N$ statistic substantially outperform all other tests on all datasets. This suggest that the proposed framework is also suitable when assessing slim-tailed distributions.

\begin{table}[htp!]
\centering
\scalebox{0.68}{
\begin{tabular}{c|c|rrrrr}
\toprule
Distr & n & JB & AD & SW & N (two-sided) & N (right-sided)\\ \midrule
 %%%%%%%%%%%%%%%%%%%%%%%%%%%%%%%%%%%%%%%%%%%%%%%
\multirow{4}{*}{Cauchy} & 20 & 86.2\% & 88.0\% & 86.6\% & 85.5\% & {\bf 89.0}\%\\
 & 50 & 99.5\% & 99.7\% & 99.6\% & 99.7\% & {\bf 99.8}\%\\
 & 100 & {\bf 100.0}\% & {\bf 100.0}\% & {\bf 100.0}\% & {\bf 100.0}\% & {\bf 100.0}\% \\
 & 250 & {\bf 100.0}\% & {\bf 100.0}\% & {\bf 100.0}\% & {\bf 100.0}\% & {\bf 100.0}\%\\ \midrule
  %%%%%%%%%%%%%%%%%%%%%%%%%%%%%%%%%%%%%%%%%%%%%%%
\multirow{4}{*}{Logistic} & 20 & 14.9\% & 10.6\% & 11.7\% & 11.1\% & {\bf 15.6}\%\\
 & 50 & 25.9\% & 15.9\% & 19.6\% & 21.2\% & {\bf 28.9}\%\\
 & 100 & 39.4\% & 23.9\% & 30.5\% & 35.8\% & {\bf 45.6}\%\\
 & 250 & 66.7\% & 46.7\% & 57.0\% & 67.8\% & {\bf 76.4}\%\\ \midrule
  %%%%%%%%%%%%%%%%%%%%%%%%%%%%%%%%%%%%%%%%%%%%%%%
\multirow{4}{*}{Student's t (2)} & 20 & 56.7\% & 52.9\% & 52.9\% & 52.9\% & {\bf 60.1}\% \\
 & 50 & 88.2\% & 85.8\% & 86.3\% & 89.2\% & {\bf 92.3}\% \\
 & 100 & 98.8\% & 98.4\% & 98.6\% & 99.2\% & {\bf 99.6}\%  \\
 & 250 & {\bf 100.0}\% & {\bf 100.0}\% & {\bf 100.0}\% & {\bf 100.0} & {\bf 100.0}\% \\ \midrule
  %%%%%%%%%%%%%%%%%%%%%%%%%%%%%%%%%%%%%%%%%%%%%%%
\multirow{4}{*}{Student's t (5)} & 20 & 22.9\% & 17.1\% & 18.6\% & 18.1\% & {\bf 23.9}\% \\
 & 50 & 43.0\% & 30.2\% & 35.5\% & 38.4\% & {\bf 46.8}\% \\
 & 100 & 64.6\% & 48.1\% & 56.4\% & 62.5\% & {\bf 70.5}\% \\
 & 250 & 92.0\% & 81.8\% & 88.2\% & 92.8\% & {\bf 95.4}\% \\ \midrule
  %%%%%%%%%%%%%%%%%%%%%%%%%%%%%%%%%%%%%%%%%%%%%%%
\multirow{4}{*}{Student's t (10)} & 20 & 12.3\% & 8.8\% & 9.8\% & 9.2\% & {\bf 12.5}\% \\
 & 50 & 20.6\% & 12.0\% & 15.5\% & 15.7\% & {\bf 21.6}\% \\
 & 100 & 30.7\% & 16.2\% & 23.0\% & 24.9\% & {\bf 33.1}\% \\
 & 250 & 52.3\% & 28.5\% & 41.7\% & 47.8\% & {\bf 57.9}\% \\ \midrule 
  %%%%%%%%%%%%%%%%%%%%%%%%%%%%%%%%%%%%%%%%%%%%%%%
\multirow{4}{*}{Student's t (20)} & 20 & 8.1\% & 6.4\% & 6.8\% & 6.5\% & {\bf 8.2}\%\\
 & 50 & 11.4\% & 7.2\% & 8.7\% & 8.4\% & {\bf 11.5}\% \\
 & 100 & 15.1\% & 8.0\% & 11.0\% & 10.7\% & {\bf 15.5}\%\\
 & 250 & 22.9\% & 10.3\% & 16.2\% & 17.2\% & {\bf 24.8}\% \\ \midrule
  %%%%%%%%%%%%%%%%%%%%%%%%%%%%%%%%%%%%%%%%%%%%%%%
\multirow{4}{*}{Student's t (30)} & 20 & 6.9\% & 5.8\% & 6.1\% & 5.9\% & {\bf 7.0}\%  \\
 & 50 & {\bf 8.9}\% & 6.2\% & 7.1\% & 6.8\% & {\bf 8.9}\% \\
 & 100 & 10.9\% & 6.5\% & 8.3\% & 7.9\% & {\bf 11.1}\%  \\
 & 250 & 15.0\% & 7.4\% & 10.7\% & 10.7\% & {\bf 16.0}\% \\ \midrule
  %%%%%%%%%%%%%%%%%%%%%%%%%%%%%%%%%%%%%%%%%%%%%%%
%\multirow{4}{*}{GN (0.5)} & 20 & 67.7\% & {\bf 75.8}\% & 70.5\% & 67.7\% & 75.5\% \\
% & 50 & 95.1\% & 98.7\% & 97.5\% & 97.7\% & {\bf 98.8}\% \\
% & 100 & 99.9\% & {\bf 100.0}\% & {\bf 100.0}\% & {\bf 100.0}\% & {\bf 100.0}\%\\
% & 250 & {\bf 100.0}\% & {\bf 100.0}\% & {\bf 100.0}\% & {\bf 100.0}\% & {\bf 100.0}\% \\ \midrule
  %%%%%%%%%%%%%%%%%%%%%%%%%%%%%%%%%%%%%%%%%%%%%%%
\multirow{4}{*}{Laplace/GN (1)} & 20 & 30.2\% & 27.2\% & 26.1\% & 26.4\% & {\bf 35.4}\% \\
 & 50 & 55.6\% & 54.5\% & 52.1\% & 59.3\% & {\bf 69.1}\%  \\
 & 100 & 79.9\% & 82.7\% & 79.7\% & 87.3\% & {\bf 92.0}\% \\
 & 250 & 98.7\% & 99.6\% & 99.2\% & 99.8\% & {\bf 99.9}\% \\ \midrule
  %%%%%%%%%%%%%%%%%%%%%%%%%%%%%%%%%%%%%%%%%%%%%%%
\multirow{4}{*}{GN (1.5)} & 20 & 12.1\% & 9.1\% & 9.5\% & 9.1\% & {\bf 13.4}\% \\
 & 50 & 19.0\% & 13.2\% & 14.4\% & 16.1\% & {\bf 23.7}\%\\
 & 100 & 27.9\% & 19.8\% & 21.5\% & 27.3\% & {\bf 37.5}\% \\
 & 250 & 48.6\% & 40.2\% & 41.4\% & 55.8\% & {\bf 66.9}\% \\ \midrule
\end{tabular}
}
\caption{The table contains {\bf test power} for various fat-tailed alternatives at significance level $\alpha=5\%$. \textit{Distr} refers to the distribution used for the alternative hypothesis and $n$ indicates sample size. The acronyms JB, AD, SW, and N refer to Jarque--Bera, Anderson--Darling, Shapiro--Wilk, and $N$ test, respectively. Best performance is marked in bold.}
\label{T:tab3}
\end{table}

\begin{table}[htp!]
\centering
\scalebox{0.70}{
\begin{tabular}{c|c|rrrr}
\toprule
Distr & n & JB & AD & SW & N (right-sided) \\ \midrule
 %%%%%%%%%%%%%%%%%%%%%%%%%%%%%%%%%%%%%%%%%%%%%%%
  %%%%%%%%%%%%%%%%%%%%%%%%%%%%%%%%%%%%%%%%%%%%%%%
\multirow{4}{*}{Logistic} & 20 &2.2\% &1.0\% &0.4\% &{\bf 4.3}\%\\
 & 50   & 2.7\% &1.2\% &0.4\% &{\bf 6.9}\%\\
 & 100 &2.7\% &1.0\% &0.3\% &{\bf 9.3}\% \\
 & 250 &1.6\% &0.5\% &0.1\% &{\bf 9.7}\% \\ \midrule
  %%%%%%%%%%%%%%%%%%%%%%%%%%%%%%%%%%%%%%%%%%%%%%%
  %%%%%%%%%%%%%%%%%%%%%%%%%%%%%%%%%%%%%%%%%%%%%%%
\multirow{4}{*}{Student's t (20)} & 20 &1.6\% &0.9\% &0.4\% &\textbf{2.8}\% \\
 & 50    &2.1\% &1.3\% &0.5\% &{\bf 3.6}\% \\
 & 100  &2.5\% &1.4\% &0.5\% &{\bf 4.7}\% \\
 & 250  &3.0\% &1.3\% &0.5\% &{\bf 7.0}\% \\ \midrule
  %%%%%%%%%%%%%%%%%%%%%%%%%%%%%%%%%%%%%%%%%%%%%%%
\multirow{4}{*}{Laplace/GN (1)} & 20 &2.1\% &1.6\% &0.2\% &{\bf 7.6}\% \\
 & 50   & 1.2\% &1.6\% &0.1\% &{\bf 9.2}\%  \\
 & 100 &0.3\% &0.7\% &0.0\% &{\bf 5.1}\%  \\
 & 250 &0.0\% &0.0\% &0.0\% &{\bf 0.2}\% \\ \midrule
  %%%%%%%%%%%%%%%%%%%%%%%%%%%%%%%%%%%%%%%%%%%%%%%
\end{tabular}
}
\caption{The table contains the {\bf unique rejections ratios} for three selected fat-tailed alternative distribution. For each test, a proportion of simulated samples on which normality hypothesis was uniquely rejected by the test (among all four considered tests) is computed. The remaining notation is aligned with Table~\ref{T:tab3}. Best performance is marked in bold.
}
\label{T:tab3b}
\end{table}

\begin{table}[htp!]
\centering
\scalebox{0.7}{
\begin{tabular}{c|c|rrrrr}
\toprule
Distr & n & JB & AD & SW & N (two-sided) & N (left-sided)\\ \midrule
 %%%%%%%%%%%%%%%%%%%%%%%%%%%%%%%%%%%%%%%%%%%%%%%
\multirow{4}{*}{GN (2.5)} & 20 & 2.4\% & 4.5\% & 4.1\% & 5.4\% & {\bf 8.6}\% \\
 & 50 & 1.3\% & 5.5\% & 4.6\% & 8.0\% & {\bf 13.6}\%\\
 & 100 & 0.9\% & 7.4\% & 6.2\% & 13.0\% & {\bf 21.1}\%\\
 & 250 & 4.3\% & 14.5\% & 12.9\% & 29.1\% & {\bf 41.3}\%\\ \midrule
 %%%%%%%%%%%%%%%%%%%%%%%%%%%%%%%%%%%%%%%%%%%%%%%
\multirow{4}{*}{GN (3)} & 20 & 1.3\% & 4.9\% & 4.4\% & 7.3\% & {\bf 12.4}\%\\
 & 50 & 0.4\% & 8.1\% & 7.0\% & 15.6\% & {\bf 24.9}\%\\
 & 100 & 0.8\% & 15.0\% & 13.7\% & 31.2\% & {\bf 44.1}\%\\
 & 250 & 22.8\% & 40.9\% & 41.7\% & 70.6\% & {\bf 81.2}\% \\  \midrule
 %%%%%%%%%%%%%%%%%%%%%%%%%%%%%%%%%%%%%%%%%%%%%%%
\multirow{4}{*}{GN (5)} & 20 & 0.5\% & 8.2\% & 7.8\% & 15.9\% & {\bf 25.4}\%\\
 & 50 & 0.1\% & 23.0\% & 24.3\% & 47.2\% & {\bf 61.2}\% \\
 & 100 & 12.3\% & 53.4\% & 60.8\% & 83.3\% & {\bf 90.8}\%\\
 & 250 & 97.0\% & 96.7\% & 99.1\% & 99.9\% & {\bf 100.0}\%\\  \midrule 
 %%%%%%%%%%%%%%%%%%%%%%%%%%%%%%%%%%%%%%%%%%%%%%%
\multirow{4}{*}{GN (10)} & 20 & 0.3\% & 13.1\% & 13.8\% & 26.4\% & {\bf 38.5}\% \\
 & 50 & 0.3\% & 43.4\% & 53.4\% & 73.0\% & {\bf 83.3}\% \\
 & 100 & 50.0\% & 84.9\% & 94.6\% & 97.7\% & {\bf 99.1}\% \\
 & 250 & {\bf 100.0}\% & {\bf 100.0}\% & {\bf 100.0}\% & {\bf 100.0}\% & {\bf 100.0}\%\\ \midrule
\end{tabular}
}\caption{The table contains {\bf test power} for various slim-tailed alternatives at significance level $\alpha=5\%$. The remaining notation is aligned with Table~\ref{T:tab3}. Best performance is marked in bold.
}
\label{T:tab3c}
\end{table}

%%%%%%%%%%%%%%%%%%%%%%%%%%%%%%%%%%%%%%%%%%%%%%%%%%%%%%%%%%%%%%%%%%%%%%%%%%%%%%%%%%%%%%%%%%%%%%%%%
\section{Mathematical framework and asymptotic results}\label{S:asymptotic}

In this section, we provide the explicit formulas for the conditional variance estimators, study their asymptotic behaviour, and show that $N$ is asymptotically normal. 

First, we introduce the basic notation and provide more explicit formulas for sets $L$, $M$, and $R$ that were given in Section~\ref{S:206020_Rule}; see \eqref{eq:LRM}.

We assume that $X\sim \mathcal{N}(\mu,\sigma)$ for mean parameter $\mu$ and standard deviation parameter $\sigma$. We use $F_X$ to denote the distribution of $X$, $\Phi$ to denote the standard normal distribution, and $\phi$ to denote the standard normal density. Following the usual convention, for any $n\in\bN$, we use $(X_1, \ldots, X_n)$ to denote the random sample from $X$ and for $i=1,\ldots, n$, we use $X_{(i)}$ to denote the sample $i$th order statistic.

For fixed partition parameters $\alpha,\beta\in \bR$, where $0\leq \alpha<\beta\leq 1$, we define the conditioning set
\begin{equation*}
A[\alpha,\beta]:=\{x\in \mathbb{R}:F_X^{-1}(\alpha)< x \leq F_X^{-1}(\beta)\}.
\end{equation*}
For brevity and with slight abuse of notation, we often write $A$ instead of  $A[\alpha,\beta]$. Then, the explicit formulas for sets $L$, $M$, and $R$ given in \eqref{eq:LRM} are
\begin{equation}\label{eq:LMR.true}
L := A[0,\tilde{q}],\quad M:= A[\tilde{q},1-\tilde{q}], \quad R:= A[1-\tilde{q},1], 
\end{equation}
where $\tilde{q}:=\Phi(x),$ and $x$ is the unique negative solution of the equation
\begin{equation}\label{eq:eq_for_tilde_q}
-x\Phi(x)-\phi(x)(1-2\Phi(x))=0.
\end{equation}
The approximate value of $\tilde{q}$ is $0.19809$; we refer to \cite[Lemma 3.3]{JawPit2015} for details. Note that \eqref{eq:eq_for_tilde_q} could be seen as a specific form of differential equation $-xy-y'(1-2y)=0$, where $y(x):=\Phi(x)$; this could be used to determine similar ratios for other distributions.

Next, we give the exact definition of the conditional sample variance. For a fixed set $A$, where $A=A[\alpha,\beta]$, the conditional variance estimator on the set $A$ is given by
\begin{equation}\label{eq:cond_var_est}
\hat{\sigma}^2_{A}:=\frac{1}{[n\beta]-[n\alpha]}\sum_{i=[n\alpha]+1}^{[n\beta]} \left(X_{(i)}-\overline{X}_{A}\right)^2,
\end{equation}
where $
[x]:=\max\{k\in \mathbb{Z}:k\leq x\}
$
denotes the floor of $x\in\bR$ and
\begin{equation}\label{eq:cond_mean_est}
\overline{X}_{A}:=\frac{1}{[n\beta]-[n\alpha]}\sum_{i=[n\alpha]+1}^{[n\beta]}X_{(i)}
\end{equation}
is the conditional sample mean. In particular, we set
$
\hat{\sigma}^2:=\hat{\sigma}^2_{A[0,1]}.
$
Recall that the test statistic $N$ is given by
\begin{equation}\label{eq:N2}
N= \frac{1}{\rho}\left(\frac{\hat\sigma^2_L-\hat\sigma^2_M}{\hat\sigma^2} +\frac{\hat\sigma^2_R-\hat\sigma^2_M}{\hat\sigma^2}\right)\sqrt{n}\,,
\end{equation}
where the normalising constant $\rho$ in \eqref{eq:N2} is approximately equal to $1.7885$; we refer to Appendix~\ref{A:rho} for the closed form formula for $\rho$. Now, we are ready to state the main result of this section, i.e. Theorem~\ref{th:asymptotic_distr_test_stat}.

\begin{theorem}\label{th:asymptotic_distr_test_stat}
Let $X\sim N(\mu,\sigma)$. Then,
\begin{equation*}
N\xrightarrow{\,d\,}\mathcal{N}(0,1)\,,\qquad n\to\infty,
\end{equation*}
where $N$ is given in \eqref{eq:N2}, and $\rho$ is a fixed normalising constant independent of $\mu$, $\sigma$, and $n$.
\end{theorem}

Before we present the proof of Theorem~\ref{th:asymptotic_distr_test_stat} let us introduce a series of Lemmas and additional notation; proof techniques are partially based on those introduced in \cite{Stigler1973}. To ease the notation, for a fixed set $A$, where $A=A[\alpha,\beta]$, we define
\begin{align*}
\mu_A & :=\bE[X|X\in A], & a & :=F_X^{-1}(\alpha)=\mu+\sigma\Phi^{-1}(\alpha),\\
\sigma^2_A &:=\bE[(X-\mu_A)^2|X\in A],& b & :=F_X^{-1}(\beta)=\mu+\sigma\Phi^{-1}(\beta),\\
\kappa_{A}&:=\tfrac{1}{(\sigma^2_A)^2}\bE[(X-\mu_A)^4|X\in A],& m_n & := [n\beta]-[n\alpha].
\end{align*}
Additionally, we set
\[
A_n:=\#\{i:X_i\leq a\}=\sum_{i=1}^n \1_{\{X_i\leq a \}},\quad
B_n:=\#\{i:X_i\leq b\}=\sum_{i=1}^n \1_{\{X_i\leq b \}},
\]
where $\1_{C}$ is the indicator function of set $C$. It is useful to note that $A_n$ and $B_n$ follow the binomial distributions $B(n,\alpha)$ and $B(n,\beta)$, respectively; note that for $\alpha=0$ and $\beta=1$ the distributions are degenerate with $A_n\equiv 0$ and $B_n \equiv n$.

Finally, for any sequence $(a_i)$ we introduce the notation of the {\it directed sum} that is given by
\begin{equation*}
\wsum_{i=k}^l a_i:=
\begin{cases}
\sum_{i=k+1}^l a_i, &\text{if }k<l,\\
0, &\text{if k=l,}\\
-\sum_{i=l+1}^k a_i, &\text{if }k>l.\\
\end{cases}
\end{equation*}

In Lemma~\ref{lm:consistency_cond_mean}, we show the consistency of the conditional sample expectation. Note that the statement of Lemma~\ref{lm:consistency_cond_mean} does not explicitly rely on normality assumption. In fact, the proof is true under very weak conditions imposed on $X$ (e.g. continuity of the distribution function of $X$); similar statement is true for other lemmas presented in this section. Also, it should be noted that Lemma~\ref{lm:consistency_cond_mean} and Lemma~\ref{lm:asymptotic_distribution_trimmed_mean} show consistency and asymptotic distribution of the standard non-parametric Expected Shortfall estimator; see e.g. \cite{McnFreEmb2010} for details.

\begin{lemma}\label{lm:consistency_cond_mean}
For any $A=A[\alpha,\beta]$, it follows that
$
\overline{X}_{A}\xrightarrow{\,\mathbb{P}\,}\mu_A,\quad n\to \infty.
$
\end{lemma}
\begin{proof}
Let $A=A[\alpha,\beta]$. For any $n\in\bN$, we get
\begin{align*}
\overline{X}_{A} & = \frac{1}{m_n}\sum_{i=A_n+1}^{B_n}X_{(i)} +\tfrac{1}{m_n}\wsum_{i=[n\alpha]}^{A_n}X_{(i)}+\tfrac{1}{m_n}\wsum_{i=B_n}^{[n\beta]}X_{(i)}.
\end{align*}
Now, we show that
\begin{equation}\label{eq:wsum1}
\tfrac{1}{m_n}\wsum_{i=[n\alpha]}^{A_n}X_{(i)}\xrightarrow{\,\mathbb{P}\,}0.
\end{equation}
Due to the consistency of the empirical quantiles, we get $X_{([n\alpha])}\xrightarrow{\mathbb{P}}a$ and $X_{(A_n)}\xrightarrow{\mathbb{P}}a$, as $n\to \infty$.
Thus, using inequality
\begin{align*}
0 & \leq  \abs{\tfrac{1}{m_n}\wsum_{i=[n\alpha]}^{A_n}X_{(i)}} \leq \abs{\frac{A_n-[n\alpha]}{m_n}} \max\{\abs{X_{([n\alpha])}},\abs{X_{(A_n)}}\},
\end{align*}
to prove \eqref{eq:wsum1}, it is sufficient to show that
$
\left| \frac{A_n-[n\alpha]}{m_n}\right|\xrightarrow{\,\mathbb{P}\,}0.
$
Noting that
\[
\frac{A_n-[n\alpha]}{m_n}=\frac{n}{m_n}\left(\frac{1}{n}A_n-\alpha\right)+\frac{n\alpha-[n\alpha]}{m_n},
\]
where
\begin{equation}\label{eq:convergence_nalpha}
\lim_{n\to\infty}\frac{n\alpha-[n\alpha]}{m_n} =0,\quad \lim_{n\to\infty}\frac{n}{m_n} =\frac{1}{\beta-\alpha},
\end{equation}
and, by the Law of Large Numbers,
$
\left(\tfrac{1}{n}A_n-\alpha\right)\xrightarrow{\,\mathbb{P}\,}0,
$
we conclude the proof of \eqref{eq:wsum1}. The proof of
\begin{equation}\label{eq:wsum2}
\tfrac{1}{m_n}\wsum_{i=B_n}^{[n\beta]}X_{(i)}\xrightarrow{\,\mathbb{P}\,}0
\end{equation}
is similar to the proof of \eqref{eq:wsum1} and is omitted for brevity. 

Next, observe that
\[
\frac{1}{m_n}\sum_{i=A_n+1}^{B_n}X_{(i)}=\frac{n}{m_n}\left(\frac{1}{n}\sum_{i=1}^{n}X_{i}\1_{\{X_i\in A\}}\right).
\]
Consequently, noting that
$
\mu_A=\frac{\bE[X\1_{\{X\in A\}}]}{\beta-\alpha},
$
and using the Law of Large Numbers, we get
\begin{equation}\label{eq:wsum3}
\frac{1}{m_n}\sum_{i=A_n+1}^{B_n}X_{(i)}\xrightarrow{\,\mathbb{P}\,}\mu_A.
\end{equation}
Combining \eqref{eq:wsum1}, \eqref{eq:wsum2}, and \eqref{eq:wsum3}, we conclude the proof. \qed
\end{proof}

Next, we focus on the asymptotic distribution of the conditional sample mean; note that Lemma~\ref{lm:asymptotic_distribution_trimmed_mean} is a slight modification of the result of \cite{Stigler1973} for trimmed means. For completeness, we present the full proof.

\begin{lemma}\label{lm:asymptotic_distribution_trimmed_mean}
For any $A=A[\alpha,\beta]$, it follows that
\begin{equation*}
\sqrt{n}\left(\overline{X}_{A}-\mu_A\right)\xrightarrow{\,d\,}\mathcal{N}(0,\eta_A),\quad n\to\infty,
\end{equation*}
for some constant $0<\eta_A<\infty$. 
\end{lemma}
\begin{proof} For brevity, we assume that $\alpha>0$ and $\beta<1.$ The remaining degenerate cases could be treated in the similar manner. Let $A=A[\alpha,\beta]$. Define
\[
S_n:=\sqrt{n}\left(\overline{X}_{A}-\mu_A\right).
\]
As in the proof of Lemma \ref{lm:consistency_cond_mean}, observe that
\begin{align}
S_n = & \frac{\sqrt{n}}{m_n}\Bigg(\sum_{i=A_n+1}^{B_n}X_{(i)}-m_n\mu_A+\wsum_{i=[n\alpha]}^{A_n}X_{(i)}+\wsum_{i=B_n}^{[n\beta]}X_{(i)}\Bigg)\nonumber\\
= & \frac{\sqrt{n}}{m_n}\Bigg(\sum_{i=A_n+1}^{B_n}\left(X_{(i)}-\mu_A\right)+(A_n-[n\alpha])(a-\mu_A)+([n\beta]-B_n)(b-\mu_A)\nonumber\\
& \quad {}+\wsum_{i=[n\alpha]}^{A_n}(X_{(i)}-a)+\wsum_{i=B_n}^{[n\beta]}(X_{(i)}-b) \Bigg).\label{eq:Sn.1}
\end{align}
Now, we show that
\begin{equation}\label{eq:2wsum1}
\tfrac{\sqrt{n}}{m_n}\wsum_{i=[n\alpha]}^{A_n}(X_{(i)}-a)\xrightarrow{\mathbb{P}}0.
\end{equation}
Due to the consistency of the empirical quantiles, we get
\[
\left(X_{([n\alpha])}-a\right)\xrightarrow{\mathbb{P}}0\quad \textrm{and}\quad  \left(X_{(A_n)}-a\right)\xrightarrow{\mathbb{P}}0,
\]
as $n\to \infty$.
Thus, using inequality
\begin{align*}
0& \leq \abs{\tfrac{\sqrt{n}}{m_n}\wsum_{i=[n\alpha]}^{A_n}(X_{(i)}-a)} \leq \abs{\frac{A_n-[n\alpha]}{m_n/\sqrt{n}}}\max\{\abs{(X_{([n\alpha])}-a},\abs{X_{(A_n)}-a}\},
\end{align*}
it is sufficient to show that $\frac{A_n-[n\alpha]}{m_n/\sqrt{n}}$ converges in distribution to some non-degenerate distribution. Note that
\[
\frac{A_n-[n\alpha]}{m_n / \sqrt{n}} =\frac{\sqrt{n}\left(\frac{1}{n}A_n-\alpha\right)}{m_n/ n}+\frac{n\alpha-[n\alpha]}{m_n / \sqrt{n}},
\]
where 
\begin{equation}\label{eq:convergence_nalpha_sqrtn}
\lim_{n\to\infty}\frac{m_n}{n} =\beta-\alpha,\quad \lim_{n\to\infty}\frac{n\alpha-[n\alpha]}{m_n / \sqrt{n}}=0,
\end{equation}
and, by the Central Limit Theorem applied to $A_n\sim B(n,\alpha)$, we get
\[
\sqrt{n}\left(\tfrac{1}{n}A_n-\alpha\right)\xrightarrow{\,d\,}\mathcal{N}\left(0,\sqrt{\alpha(1-\alpha)}\right).
\]
Thus, using the Slutsky's Theorem (see e.g. \cite[Theorem 6']{Ferguson1996}), we get
\[
\frac{A_n-[n\alpha]}{m_n / \sqrt{n}}\xrightarrow{\,d\,} \mathcal{N}\left(0,\sqrt{\frac{\alpha(1-\alpha)}{\beta-\alpha}}\right),
\]
which concludes the proof of \eqref{eq:2wsum1}. Similarly, one can show that
\begin{equation}\label{eq:2wsum2}
\tfrac{\sqrt{n}}{m_n}\wsum_{i=B_n}^{[n\beta]}(X_{(i)}-b)\xrightarrow{\mathbb{P}}0.
\end{equation}
Combining \eqref{eq:2wsum1} with \eqref{eq:2wsum2}, and noting that
 \begin{equation}\label{eq:floor.nofloor}
\frac{n\alpha-[n\alpha]}{m_n / \sqrt{n}}\to 0\quad\textrm{and}\quad \frac{[n\beta]-n\beta}{m_n / \sqrt{n}}\to 0,
 \end{equation}
we can rewrite \eqref{eq:Sn.1} as
\[
S_n =  \frac{\sqrt{n}}{m_n}\bigg(\sum_{i=A_n+1}^{B_n}(X_{(i)}-\mu_A)+(A_n-n\alpha)(a-\mu_A)+ (n\beta-B_n)(b-\mu_A) \bigg)+r_n,
\]
where $r_n\xrightarrow{\mathbb{P}}0$. Next, we have
\[
S_n=\frac{n(\beta-\alpha)}{m_n}\left(\frac{\sqrt{n}}{n(\beta-\alpha)} \sum_{i=1}^n Z^A_i\right)+r_n,
\]
where for $i=1,\ldots,n$, we set
\begin{align*}
Z^A_i & :=(X_i-\mu_A)\1_{\{X_i\in A\}}+(\1_{\{X_i\leq a\}}-\alpha)(a-\mu_A)+(\beta-\1_{\{X_i\leq b\}})(b-\mu_A).
\end{align*}
Finally, noting that for $n\to\infty$ we get
$
\frac{n(\beta-\alpha)}{m_n}\xrightarrow{\,\mathbb{P}\,}1,
$
and combining the Central Limit Theorem applied to $(Z^A_i)$ with the Slutsky's Theorem we conclude the proof; note that $(Z^A_i)$ are i.i.d. with zero mean and finite variance. \qed
\end{proof}

Next, we show that for the conditional variance estimator one can substitute the sample mean with the true mean without impacting the asymptotics. For any $A$, where $A=A[\alpha,\beta]$, the conditional variance estimator with known mean is given by
\begin{equation*}
\hat{s}^2_{A}:=\frac{1}{m_n}\sum_{i=[n\alpha]+1}^{[n\beta]} \left(X_{(i)}-\mu_A\right)^2.
\end{equation*}

\begin{lemma}\label{lm:equiv_asympt_distr}
For any $A=A[\alpha,\beta]$, it follows that
$\sqrt{n}\left(\hat{\sigma}^2_{A}-\hat{s}^2_{A}\right)\xrightarrow{\,\mathbb{P}\,}0$, $n\to\infty$.
\end{lemma}
\begin{proof}
As in the proof of Lemma \ref{lm:asymptotic_distribution_trimmed_mean}, we focus on the case $0<\alpha<\beta<1$. Let $A=A[\alpha,\beta]$ and note that
\begin{align*}
\hat{s}^2_{A} & = \frac{1}{m_n} \sum_{i=[n\alpha]+1}^{[n\beta]} \left(X_{(i)}-\overline{X}_{A}+\overline{X}_{A}-\mu_A\right)^2\\
& =  \frac{1}{m_n}\sum_{i=[n\alpha]+1}^{[n\beta]} \left(X_{(i)}-\overline{X}_{A}\right)^2  +\left(\overline{X}_{A}-\mu_A\right)^2 \\
& \phantom{=} +\frac{2}{m_n}\left(\overline{X}_{A}-\mu_A\right)\sum_{i=[n\alpha]+1}^{[n\beta]}\left(X_{(i)}-\overline{X}_{A}\right),
\end{align*}
where the last summand equals $0$ since 
\[
\sum_{i=[n\alpha]+1}^{[n\beta]}X_{(i)}=m_n\overline{X}_{A}=\sum_{i=[n\alpha]+1}^{[n\beta]}\overline{X}_{A} \\.
\]
Consequently, we get
$
\sqrt{n}\left(\hat{s}^2_{A}-\hat{\sigma}^2_{A}\right) =\sqrt{n}\left(\overline{X}_{A}-\mu_A\right)^2.
$
Thus, using Lemma \ref{lm:consistency_cond_mean} combined with Lemma \ref{lm:asymptotic_distribution_trimmed_mean}, we conclude the proof. \qed
\end{proof}

Now, we study the asymptotic behaviour of the conditional variance estimator; this is a key lemma that will be used in the proof of Theorem~\ref{th:asymptotic_distr_test_stat}. Moreover, this result may be of independent interest since it allows one to construct the asymptotic confidence interval for the conditional variance.

\begin{lemma}\label{lm:asymptotic_distribution_sigma2M}
For any $A=A[\alpha,\beta]$ it follows that
\begin{equation*}
\sqrt{n}\left(\hat{\sigma}^2_{A}-\sigma^2_A\right)\xrightarrow{d}\mathcal{N}(0,\tau_A),
\end{equation*}
where
\begin{align}
\tau^2_A & :=  \frac{1}{(\beta-\alpha)^2}\Big((\beta-\alpha)(\sigma^2_A)^2(\kappa_A-1) \nonumber\\
& \quad {} +\alpha(1-\alpha)\left((a-\mu_A)^2-\sigma^2_A\right)^2 +\beta(1-\beta)\left((b-\mu_A)^2-\sigma^2_A\right)^2 \nonumber \\
& \quad {} -2\alpha(1-\beta)\left((a-\mu_A)^2-\sigma^2_A\right)\left((b-\mu_A)^2-\sigma^2_A\right)\Big).\label{eq:tau.A}\footnotemark
\end{align}
\end{lemma}
\footnotetext{Note that for degenerate cases $\alpha=0$  and  $\beta=1$, we get $a=-\infty$ and $b=\infty$, respectively. In those cases, the convention $0\cdot \infty =0$ should be used.}
\begin{proof}
Due to Lemma \ref{lm:equiv_asympt_distr}, it is enough to consider $\hat{s}^2_{A}$ instead of $\hat{\sigma}^2_{A}$. For
\[
S_n^A:=\sqrt{n}\left(\hat{s}^2_{A}-\sigma^2_A\right),
\]
we get
\begin{align}\label{eq:Sn.l5}
S_n^A & =  \frac{\sqrt{n}}{m_n} \Bigg(\sum_{i=A_n+1}^{B_n} \left( \left( X_{(i)}-\mu_A \right)^2-\sigma^2_A \right)\nonumber \\
&\quad {} \quad {} + (A_n-[n\alpha])\left((a-\mu_A)^2-\sigma^2_A\right) + ([n\beta]-B_n)\left((b-\mu_A)^2-\sigma^2_A\right) \nonumber\\
&\quad {} + \wsum_{i=[n\alpha]}^{A_n} \left( \left( X_{(i)}-\mu_A \right)^2-\left( a-\mu_A \right)^2 \right) \nonumber\\
&\quad {} + \wsum_{i=B_n}^{[n\beta]} \left( \left( X_{(i)}-\mu_A \right)^2-\left( b-\mu_A \right)^2 \right)\Bigg).
\end{align}
By the arguments similar to the ones presented in the proof of Lemma \ref{lm:asymptotic_distribution_trimmed_mean}, we get
\[
\tfrac{\sqrt{n}}{m_n}\wsum_{i=[n\alpha]}^{A_n} \left( \left( X_{(i)}-\mu_A \right)^2-\left( a-\mu_A \right)^2 \right)\xrightarrow{\,\mathbb{P}\,}0,
\]
and
\[
\tfrac{\sqrt{n}}{m_n}\wsum_{i=B_n}^{[n\beta]} \left( \left( X_{(i)}-\mu_A \right)^2-\left( b-\mu_A \right)^2 \right)\xrightarrow{\,\mathbb{P}\,}0.
\]
Thus, recalling \eqref{eq:floor.nofloor}, we can rewrite \eqref{eq:Sn.l5} as
\begin{align*}
S_n^A = & \frac{\sqrt{n}}{m_n} \Bigg(\sum_{i=A_n+1}^{B_n} \left( \left( X_{(i)}-\mu_A \right)^2-\sigma^2_A \right) \nonumber\\
& + (A_n-n\alpha)\left((a-\mu_A)^2-\sigma^2_A\right) + (n\beta-B_n)\left((b-\mu_A)^2-\sigma^2_A\right)\Bigg)+r_n,
\end{align*}
where $r_n\xrightarrow{\mathbb{P}}0$. Next, for $i=1,\ldots,n$, we set
\begin{align}
Y^A_i & :=  \left( \left( X_{i}-\mu_A \right)^2-\sigma^2_A \right) \1_{\{X_i\in A\}}\nonumber\\
& \quad {}+\left(\1_{\{X_i\leq a\}}-\alpha\right)\left((a-\mu_A)^2-\sigma^2_A\right) \nonumber\\
& \quad {} +\left(\beta-\1_{\{X_i\leq b\}}\right)\left((b-\mu_A)^2-\sigma^2_A\right)\label{eq:YA},
\end{align}
and by straightforward computations we get
$
\bE[Y^A_i]=0$ and $D^2[Y^A_i]=(\beta-\alpha)^2 \tau^2_A$. Consequently, noting that
\begin{equation}\label{eq:S_n}
S_n^A= \frac{\sqrt{n}}{m_n} \sum_{i=1}^n Y^A_i + r_n,
\end{equation}
and using the Central Limit Theorem combined with the Slutsky's Theorem, we conclude the proof. \qed
\end{proof}
Finally, we are ready to show the proof of Theorem~\ref{th:asymptotic_distr_test_stat}.

\begin{proof}[of Theorem~\ref{th:asymptotic_distr_test_stat}]
For conditioning sets $L$, $M$, and $R$ given in \eqref{eq:LMR.true}, we define the associated sequences of random variables $(Y_i^L)$, $(Y_i^M)$, $(Y_i^R)$ using \eqref{eq:YA}. For any $n\in\bN$, we set 
\[
Z_n:=\sqrt{n}\begin{bmatrix}
\frac{1}{n}\sum_{i=1}^n \frac{1}{\tilde{q}}Y_i^L \\
\frac{1}{n}\sum_{i=1}^n \frac{1}{1-2\tilde{q}}Y_i^M \\
\frac{1}{n}\sum_{i=1}^n \frac{1}{\tilde{q}}Y_i^R
\end{bmatrix},
\]
where $\tilde{q}$ is defined via \eqref{eq:eq_for_tilde_q}. By the multivariate Central Limit Theorem  (cf. \cite[Theorem 5]{Ferguson1996}), we get
$
Z_n\xrightarrow{\,d\,}\mathcal{N}_3(0,\Sigma),
$
where
\[
\Sigma:= \begin{bmatrix}
\frac{\Cov(Y_1^L,Y_1^L)}{\tilde{q}^2}  & \frac{\Cov(Y_1^M,Y_1^L)}{\tilde{q}(1-2\tilde{q})} & \frac{\Cov(Y_1^R,Y_1^L)}{\tilde{q}^2} \\
\frac{\Cov(Y_1^L,Y_1^M)}{\tilde{q}(1-2\tilde{q})} & \frac{\Cov(Y_1^M,Y_1^M)}{(1-2\tilde{q})^2}  & \frac{\Cov(Y_1^R,Y_1^M)}{\tilde{q}(1-2\tilde{q})} \\
\frac{\Cov(Y_1^L,Y_1^R)}{\tilde{q}^2} & \frac{\Cov(Y_1^M,Y_1^R)}{\tilde{q}(1-2\tilde{q})} & \frac{\Cov(Y_1^R,Y_1^R)}{\tilde{q}^2} 
\end{bmatrix}.
\]
Now, let
$
S_n :=\sqrt{n}\left(\hat{\sigma}^2_L+\hat{\sigma}^2_R-2\hat{\sigma}^2_M\right).
$
Using \eqref{rq:LMR}, it is easy to see that
\begin{equation}\label{eq:Sn.final}
S_n=\sqrt{n}\left(\hat{\sigma}^2_L-\sigma^2_L+\hat{\sigma}^2_R-\sigma^2_R-2\hat{\sigma}^2_M+2\sigma^2_M\right).
\end{equation}
Consequently, by the arguments similar to the ones presented in the proof of Lemma \ref{lm:asymptotic_distribution_sigma2M} (see \eqref{eq:S_n}), we can rewrite \eqref{eq:Sn.final} as
$
S_n=M_n Z_n +r_n,
$
where $r_n\xrightarrow{\,\mathbb{P}\,}0$ and
\[
M_n:=\begin{bmatrix}
\frac{n\tilde{q}}{[n\tilde{q}]}, & -2\frac{n(1-2\tilde{q})}{[n(1-\tilde{q})]-[n\tilde{q}]}, & \frac{n\tilde{q}}{n-[n(1-\tilde{q})]}
\end{bmatrix}.
\]
Next, observing that $M_n\xrightarrow{\mathbb{P}} [1,-2,1]$ and using the multivariate Slutsky's Theorem (cf. \cite[Theorem 6]{Ferguson1996}), we get
$
S_n\xrightarrow{\,d\,} \mathcal{N}(0,\tau),
$
where
\begin{equation}\label{eq:tau}
\tau:=\sqrt{[1,-2,1] \,\Sigma\, [1,-2,1]^T}.
\end{equation}
Let
$
\rho:=\frac{\tau}{\sigma^2}
$
and
$
N_n:=\frac{1}{\rho}\frac{S_n}{\hat \sigma^2}.
$
Observing that $\sigma^2 / \hat{\sigma}^2_n\xrightarrow{\,\mathbb{P}\,}1$, and again using the Slutsky's Theorem, we get
$
N_n\xrightarrow{\,d\,} \mathcal{N}(0,1).
$
To conclude the proof of Theorem~\ref{th:asymptotic_distr_test_stat}, we need to show that $\rho$ is independent of $\mu$ and $\sigma$.

To do so, let us first show that for any $A$, where $A=A[\alpha,\beta]$, and the corresponding random variable $Y^A_1$ given in \eqref{eq:YA}, we get
\begin{equation}\label{eq:Y.norm}
Y^A_1 = \sigma^2 \psi(\tilde X_1,\alpha,\beta)\,,
\end{equation}
where $\tilde X_1 := (X_1-\mu)/\sigma$ and
$
\psi\colon \bR\times [0,1]\times [0,1] \to \bR
$
is some fixed measurable function. From \cite[Section 13.10.1]{JohKotBal1994}, we know that
\begin{align}
\mu_A & = \sigma \frac{\phi(\Phi^{-1}(\alpha))-\phi(\Phi^{-1}(\beta))}{\beta-\alpha}+\mu,\label{eq:explicit_mu_A}\\
\sigma^2_A & = \sigma^2\left(\frac{\Phi^{-1}(\alpha)\phi(\Phi^{-1}(\alpha))-\Phi^{-1}(\beta)\phi(\Phi^{-1}(\beta))}{\beta-\alpha}\right. \nonumber\\
& \phantom{=} \left. -\frac{\left(\phi(\Phi^{-1}(\alpha))-\phi(\Phi^{-1}(\beta))\right)^2}{(\beta-\alpha)^2}+1\right).\label{eq:explicit_sigma_A}\footnotemark
\end{align}
\footnotetext{For $\alpha=0$ or $\beta=1$ the convention $0\cdot \pm\infty =0$ is used.}
Consequently, the standardised mean $\tilde\mu_A :=\frac{\mu_A-\mu}{\sigma}$ and variance $\tilde\sigma_A^2:=\frac{\sigma^2_A}{\sigma^2}$ depend only on $\alpha$ and $\beta$. Recalling \eqref{eq:YA}, we get
\begin{align}
\frac{Y^A_1}{\sigma^2} & =  \left( \left( \frac{X_1-\mu_A}{\sigma} \right)^2-\tilde\sigma^2_A \right) \1_{\{X_1\in A\}} +\left(\1_{\{X_1\leq a\}}-\alpha\right)\left(\left(\frac{a-\mu_A}{\sigma}\right)^2-\tilde\sigma^2_A\right) \nonumber \\
& \quad {}  +\left(\beta-\1_{\{X_1\leq b\}}\right)\left(\left(\frac{b-\mu_A}{\sigma}\right)^2-\tilde\sigma^2_A\right) \nonumber \\
& =\left( \left(\tilde X_1-\tilde\mu_A \right)^2-\tilde\sigma^2_A \right) \1_{\{X_i\in A\}} +\left(\1_{\{X_1\leq a\}}-\alpha\right)\left(\left(\Phi^{-1}(\alpha)-\tilde\mu_A\right)^2-\tilde\sigma^2_A\right)  \nonumber \\
& \quad {} +\left(\beta-\1_{\{X_1\leq b\}}\right)\left(\left(\Phi^{-1}(\beta)-\tilde\mu_A\right)^2-\tilde\sigma^2_A\right).
\end{align}
Combining this with equalities
\begin{align*}
\{X_1\in A\}&= \{ \tilde X_1 \in [\Phi^{-1}(\alpha), \Phi^{-1}(\beta)) \},\\
\{X_1\leq a\}&= \{ \tilde X_1 \leq \Phi^{-1}(\alpha) \},\\
\{X_1\leq b\}&= \{ \tilde X_1 \leq \Phi^{-1}(\beta) \},
\end{align*}
we conclude the proof of \eqref{eq:Y.norm}.

Now, using \eqref{eq:Y.norm} for $L$, $M$, and $R$, and expressing $\Sigma / \sigma^4$  as
\[
\begin{bmatrix}
\frac{\Cov\left(\frac{Y_1^L}{\sigma^2},\frac{Y_1^L}{\sigma^2}\right)}{\tilde{q}^2}  & \frac{\Cov\left(\frac{Y_1^M}{\sigma^2},\frac{Y_1^L}{\sigma^2}\right)}{\tilde{q}(1-2\tilde{q})} & \frac{\Cov\left(\frac{Y_1^R}{\sigma^2},\frac{Y_1^L}{\sigma^2}\right)}{\tilde{q}^2} \\
\frac{\Cov\left(\frac{Y_1^L}{\sigma^2},\frac{Y_1^M}{\sigma^2}\right)}{\tilde{q}(1-2\tilde{q})} & \frac{\Cov\left(\frac{Y_1^M}{\sigma^2},\frac{Y_1^M}{\sigma^2}\right)}{(1-2\tilde{q})^2}  & \frac{\Cov\left(\frac{Y_1^R}{\sigma^2},\frac{Y_1^M}{\sigma^2}\right)}{\tilde{q}(1-2\tilde{q})} \\
\frac{\Cov\left(\frac{Y_1^L}{\sigma^2},\frac{Y_1^R}{\sigma^2}\right)}{\tilde{q}^2} & \frac{\Cov\left(\frac{Y_1^M}{\sigma^2},\frac{Y_1^R}{\sigma^2}\right)}{\tilde{q}(1-2\tilde{q})} & \frac{\Cov\left(\frac{Y_1^R}{\sigma^2},\frac{Y_1^R}{\sigma^2}\right)}{\tilde{q}^2} 
\end{bmatrix},
\]
\noindent we see that $\Sigma / \sigma^4$ does not depend on $\mu$ and $\sigma$.

Finally, recalling \eqref{eq:tau} and the definition of $\rho$ we conclude the proof of Theorem~\eqref{th:asymptotic_distr_test_stat}; we refer to Appendix~\ref{A:rho} for the closed-form formula for $\rho$. \qed
\end{proof}

The results presented in this section could be directly applied to various other non-parametric quantile estimators and to the unbiased variance estimators.\footnote{In particular, this refers to the estimator implemented via {\it quantile} function in {\bf R} that was used in Figure~\ref{fig:Rcode}.} This is summarised in the next two remarks.
\begin{remark}\label{rmk:unbiased_var_est}
The standard formula for the whole sample (unbiased) variance uses $n-1$ instead of $n$ in the denominator. In the conditional case, this would be reflected in the different formula for \eqref{eq:cond_var_est}, where  $m_n$ is replaced by $m_n-1$. Note that the statement of Theorem \ref{th:asymptotic_distr_test_stat} remains valid for the modified conditional variance estimator due to the combination of the Slutsky's Theorem and the fact that $(m_n-1)/m_n\to 1$.
\end{remark}

\begin{remark}\label{rmk:other_quantile_est}
When defining the conditional sample variance \eqref{eq:cond_var_est}, we used $[n\alpha]+1$ and $[n\beta]$ as the limits of the summation in \eqref{eq:cond_var_est} and \eqref{eq:cond_mean_est}. This choice corresponds to the non-parametric $\alpha$-quantile estimator given by $X_{([n\alpha])}$.

In the literature there exist many different formulas for non-parametric quantile estimators, most of which are bounded by the nearest order statistics; see \cite{HyndFan1996} for details. It is relatively easy to show that all results presented in this section hold true if we replace $[n\alpha]$ and $[n\beta]$ by suitably chosen sequences that correspond to different empirical quantile choices. For completeness, we provide a more detailed description of this statement.

Consider sequences $(\alpha_n)$ and $(\beta_n)$ such that $n\alpha-\alpha_n$ and $\beta_n-n\beta$ are bounded, and define
$\tilde{m}_n:=\beta_n-\alpha_n$. The corresponding conditional sample mean and variance are given by
\[
\bar{X}_{A}^{\ast} :=\frac{1}{\tilde{m}_n}\sum_{i=\alpha_n+1}^{\beta_n}X_{(i)}\quad\textrm{and}\quad \hat{\sigma}^{2,\ast}_{A} :=\frac{1}{\tilde{m}_n}\sum_{i=\alpha_n+1}^{\beta_n} \left(X_{(i)}-\bar{X}_{A}^{\ast}\right)^2.
\]
Then, we can replace $\overline{X}_A$ and $\hat{\sigma}^2_A$ by $\bar{X}_{A}^{\ast}$ and $\hat{\sigma}^{2,\ast}_{A}$ in Theorem~\ref{th:asymptotic_distr_test_stat} as well as in all lemmas presented in the section.

Instead of showing a full proof, we briefly comment how to show consistency of quantile estimators as well as comment on counterparts of \eqref{eq:convergence_nalpha} and \eqref{eq:convergence_nalpha_sqrtn}. All proofs could be translated using a very similar logic.

First, noting that for some $k\in \bN$ we get $X_{([n\alpha]-k)}\leq X_{(\alpha_n)}\leq X_{([n\alpha]+k)}$ and 
$X_{([n\beta]-k)}\leq X_{(\beta_n)}\leq X_{([n\beta]+k)}$,
it is straightforward to check that $X_{(\alpha_n)}$ and $X_{(\beta_n)}$ are consistent $\alpha$-quantile and $\beta$-quantile estimators; see e.g. \cite[Section 2.3]{Serfling1980}. Second, to show the analogue of \eqref{eq:convergence_nalpha}, it is enough to use the boundedness of $n\alpha-\alpha_n$ and $\beta_n-n\beta$, and note that $\frac{n\alpha-\alpha_n}{n}\to 0$ and $\frac{\beta_n-n\beta}{n}\to 0$. Third, to show \eqref{eq:convergence_nalpha_sqrtn}, it is enough to use boundedness of $n\alpha-\alpha_n$ and note that for some $k\in\bN$ we get
$
\frac{\abs{n\alpha-\alpha_n}}{\tilde{m}_n/\sqrt{n}}\leq \frac{k}{\tilde{m}_n}\sqrt{n}=\frac{\tilde{m}_n}{n}\frac{k}{\sqrt{n}}
$.
\end{remark}

%%%%%%%%%%%%%%%%%%%%%%%%%%%%%%%%%%%%%%%%%%%%%%%%%%%%%%%%%%%%%%%%%%%%%%%%%%%%%%%%%%%%%%%%%%%%%%%%%

\section{Empirical example: case study of market stock returns}\label{S:2}
In this section we apply the proposed framework to stock market returns performing a basic sanity-check verification. Before we do that, let us comment on the connection between the 20-60-20 Rule and financial time series. 

Assuming that $X$ describes financial asset return rates we can split the population using 20/60/20 ratio and check the behaviour of returns within each subset. If non-normal perturbations are observed only for extreme events, the 20/60/20 break might identify the regime switch and provide a good spatial clustering; this could be linked to a popular financial \textit{stylised fact} saying that average financial asset returns tend to be normal, but the extreme returns are not -- see \cite{Cont2001} and \cite{SheQia2010} for details. It should be emphasized that according to authors' best knowledge, the link between this property and data non-normality was not discussed in the literature before.

The easiest way to verify this hypothesis is to take stock return samples for different periods, make the quantile--quantile plots (with standard normal as a reference distribution) and check if the clustering is accurate. In Figure~\ref{F:fig1}, we present exemplary results for two major US stocks, namely GOOGL and AAPL, and two major stock indices, namely S\&P500 and DAX; we took time-series of length 250 for different time intervals ranging in the period from 10/2015 to 01/2018.\footnote{Data is downloaded from \href{https://finance.yahoo.com/}{Yahoo Finance} via {\bf R} \textit{tidyquant} package.}

\begin{figure}[htp!]
\begin{center}
\includegraphics[width=0.35\textwidth]{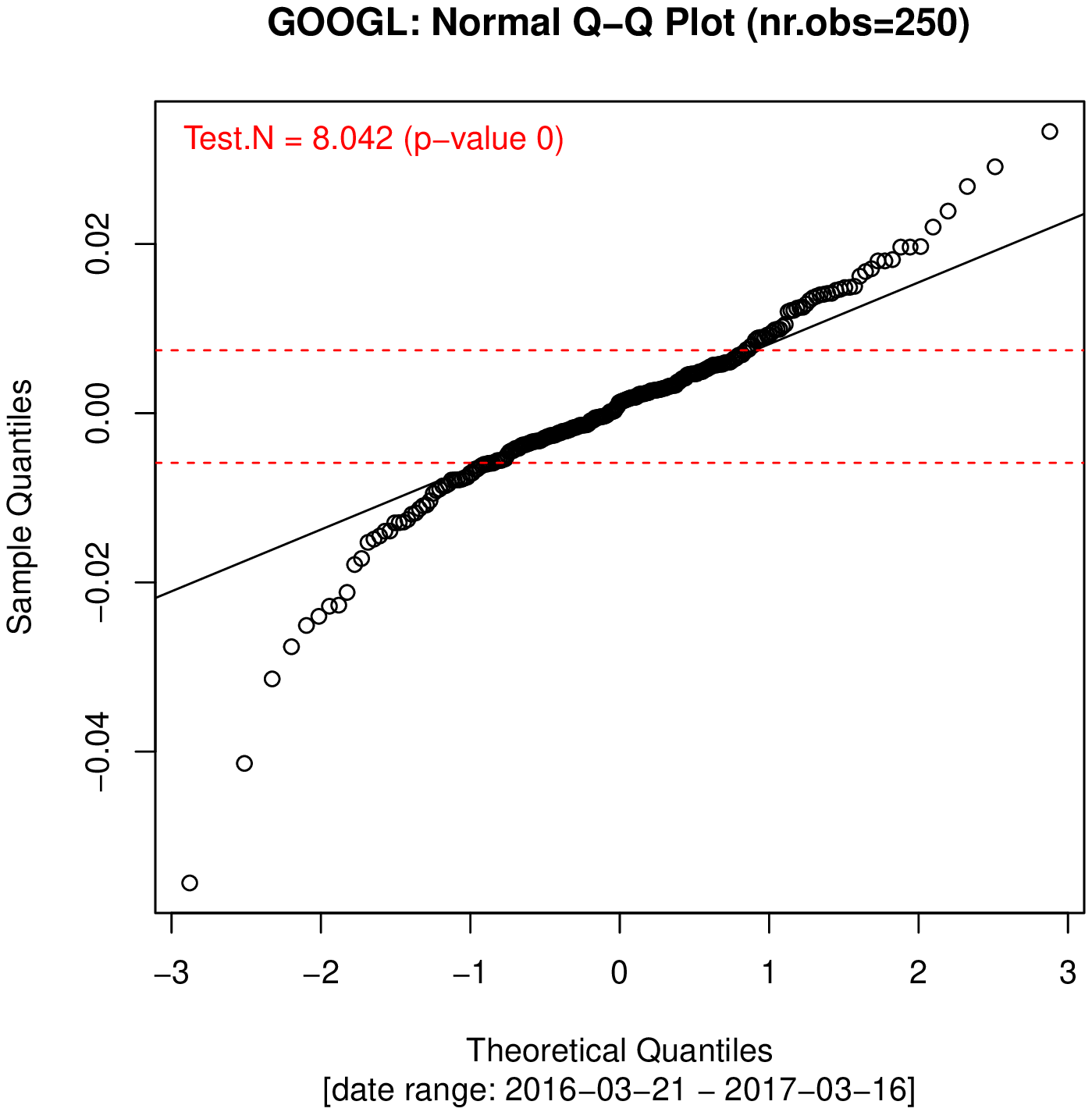}
\includegraphics[width=0.35\textwidth]{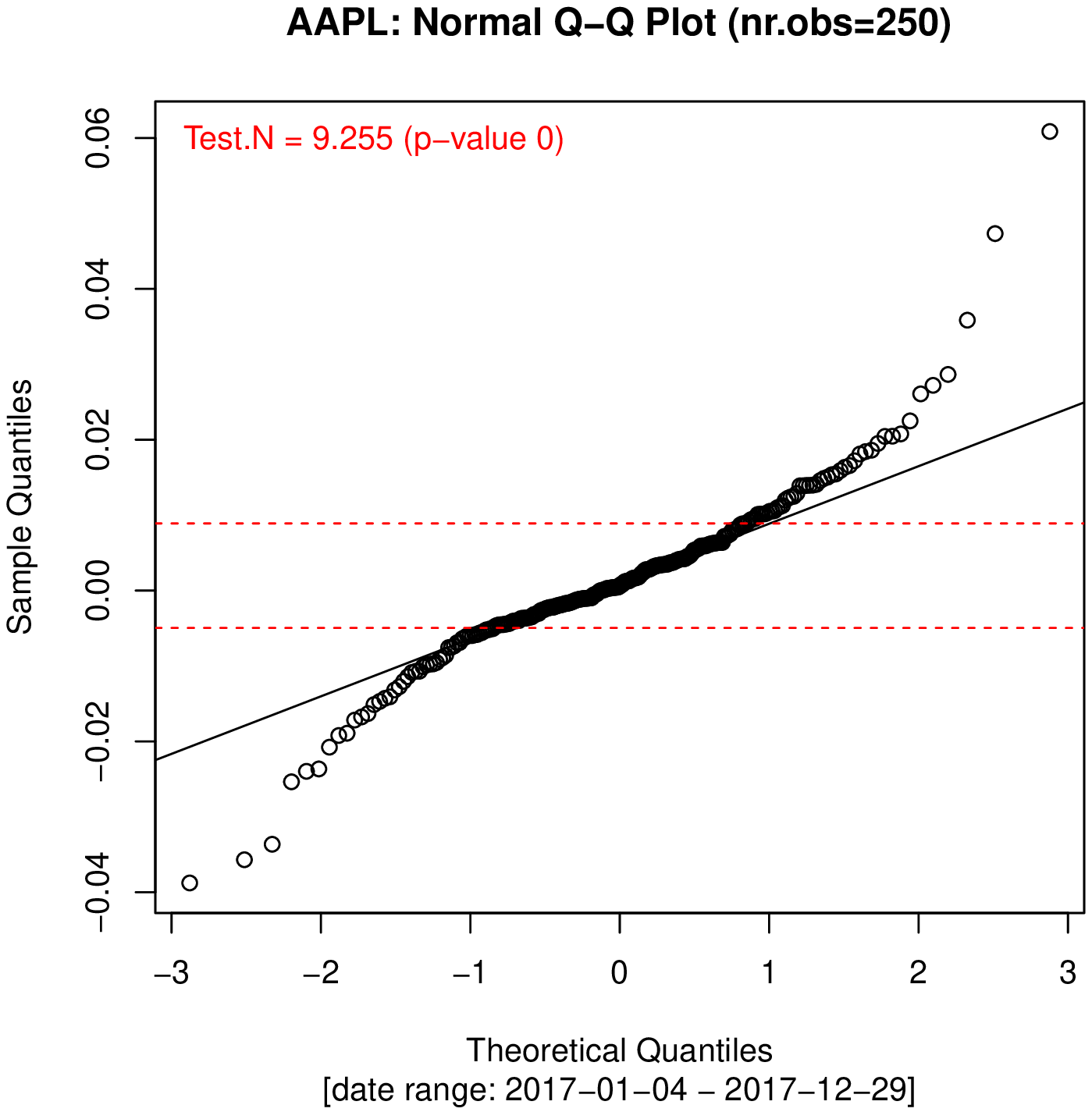}\\
\includegraphics[width=0.35\textwidth]{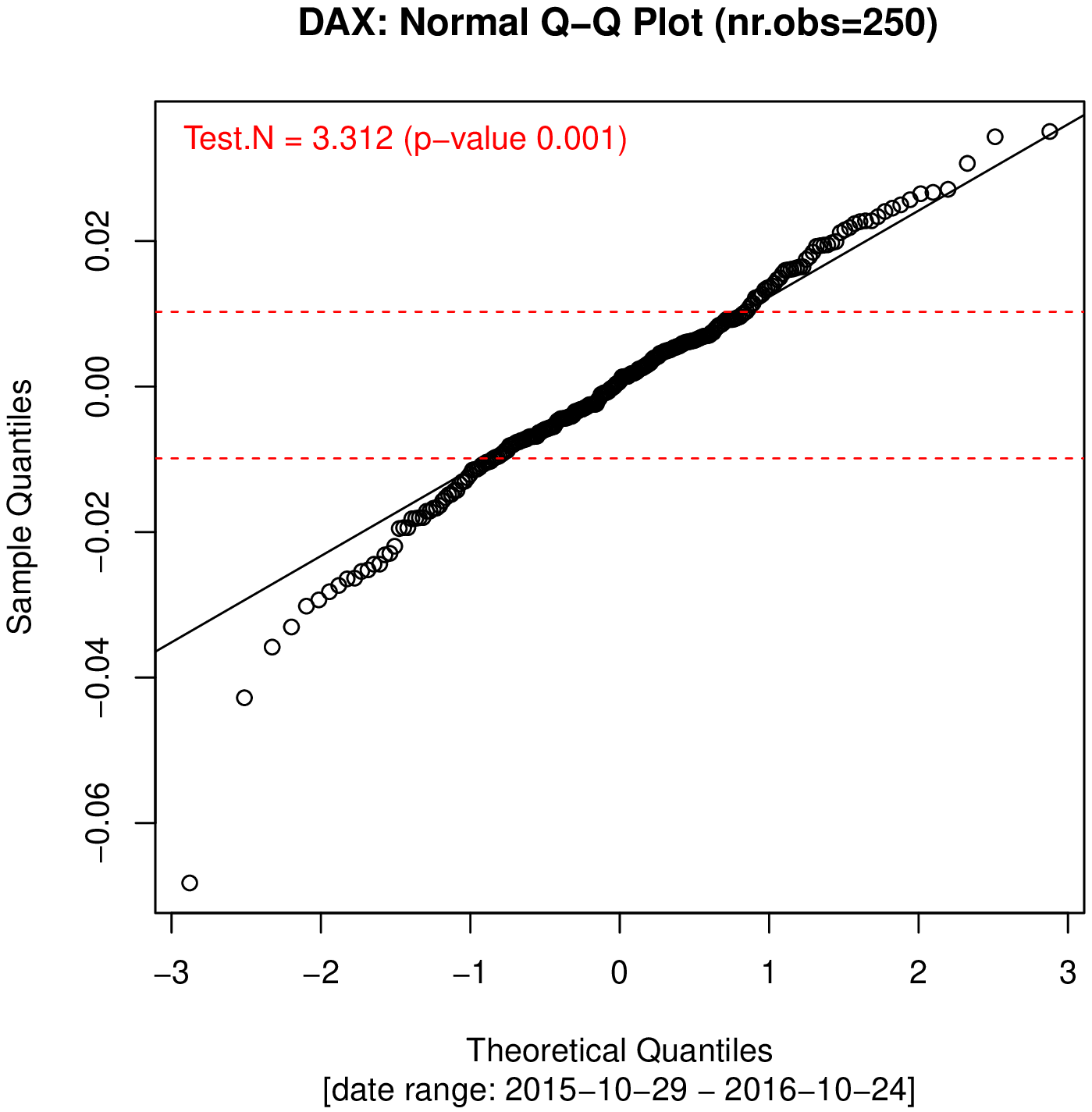}
\includegraphics[width=0.35\textwidth]{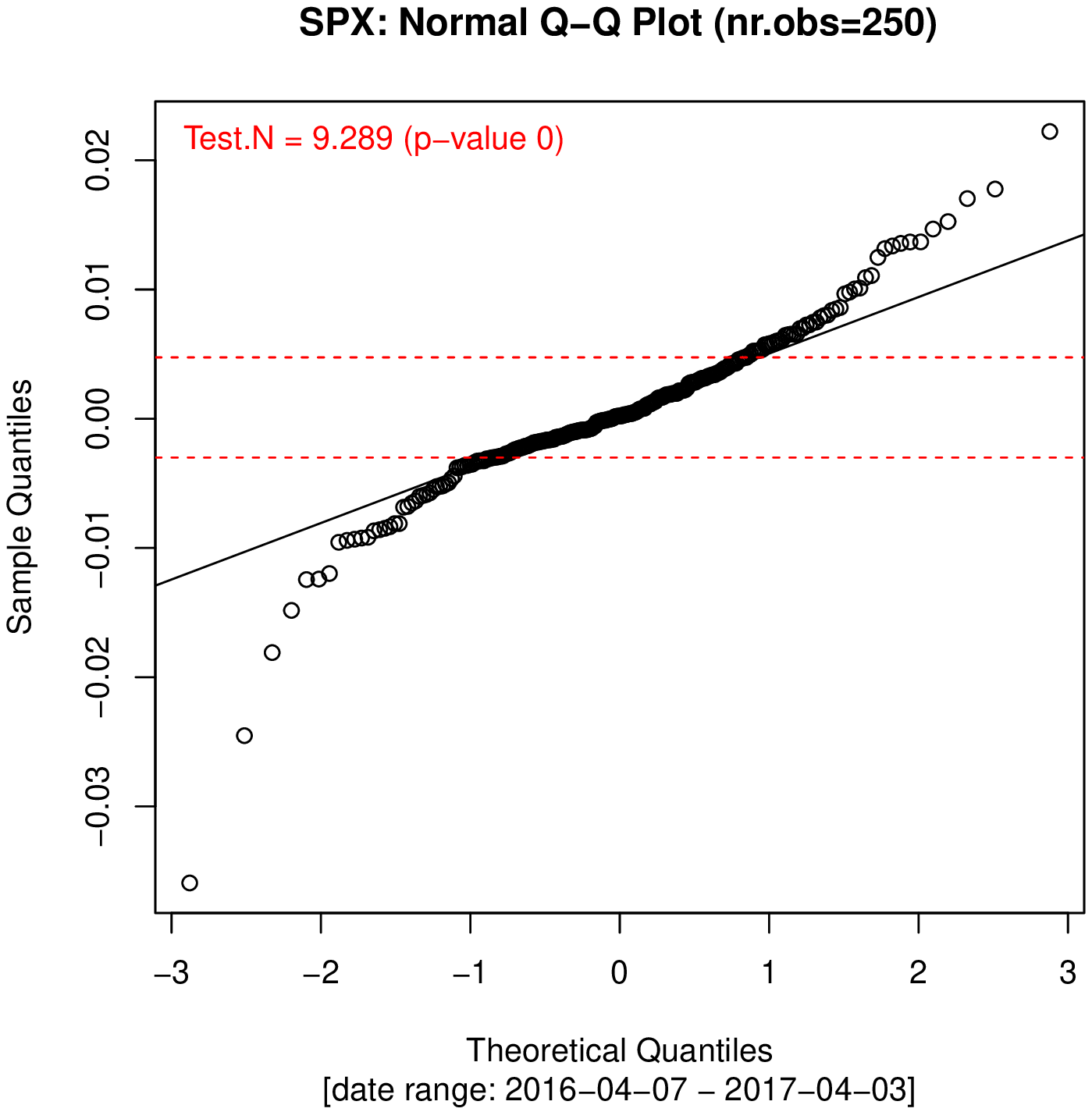}
\end{center}
\caption{Quantile-Quantile plots for various financial returns of size $n=250$. The dashed lines correspond to 20\% and 80\% quantiles. One can see different data behaviour in each cluster: while the central region is aligned with normal distribution the tail regions are not.}
\label{F:fig1}
\vspace{-0.3cm}
\end{figure}

From Figure~\ref{F:fig1} we see that this division is surprisingly accurate: a very good normal fit is observed in the $M$ set (middle 60\% of observations), while the fit in the tail sets $L$ and $R$ (bottom and top $20\%$ of observations) is bad.  By taking different sample sizes, different time-horizons, and different stocks we can confirm that this property is systematic, i.e. the results are almost always similar to the ones presented in Figure~\ref{F:fig1}.

While the presence of fat-tails in asset return distributions is a well-known observation in the financial world, it is quite surprising to note that the non-normal behaviour could be seen for approximately 40\% of data. Also, test statistic $N$ can be used to formally quantify this phenomenon and to measure tail heaviness: the bigger the conditional standard deviation in the tails (in reference to the central part), the fatter the tails.

In the following, we focus on assessing the performance of the test statistic $N$ on market data. We perform a simple empirical study and take returns of all stocks listed in S\&P500 index on 16.06.2018 that have full historical data in the period from 01.2000 to 05.2018. This way we get full data (4610 daily adjusted close price returns) for 381 stocks. Next, for a given sample size $n\in\{50, 100, 250\}$ we split the returns into disjoint sets of length $n$, and for each subset we compare the value of $N$ with the corresponding empirical quantiles presented in Figure~\ref{F:fig2}. More precisely, using $N$ we perform a right-sided statistical test and reject normality (null) hypothesis if the computed value is greater than the empirical value $F^{-1}_{n}(1-\alpha)$, for $\alpha\in\{1\%, 2.5\%, 5\%\}$.

To assess test performance, we compare the results with other benchmark normality tests: Jarque--Bera test, Anderson--Darling test, and Shapiro--Wilk test. While the non-normality of returns is a well-known fact, and all testing frameworks should show good performance, we want to check if our framework leads to some new interesting results. We check the normality hypothesis and compute two supplementary metrics that are used for performance assessment:
\begin{enumerate}[-]
\item {\bf Statistic T} gives the {\bf total rejection ratio} of a given test. It corresponds to the proportion of data on which the normality assumption was rejected at a given significance level; it is the ratio of rejected data subsamples to all data subsamples.
\item {\bf Statistic U} gives the {\bf unique rejection ratio} of a given test. It corresponds to the proportion of data on which normality assumption was rejected at a given significance level only by the considered test (among all four tests); it is the ratio of uniquely rejected data subsamples to all data subsamples.
\end{enumerate}

\noindent The combined results for all values of $n$ and $\alpha$ are presented in Table~\ref{T:tab1}. 
\begin{table}[htp!]
\centering
\scalebox{0.8}{
\begin{tabular}{l|l|c|c|c|rrrr}\toprule
Desc & nr runs& $\alpha$ & n &rejects&  JB& AD & SW &N\\ \midrule
%%%%%%%%%%%%%%%%%%%%%%%%%%%%%%%%%%%%%%%%%%%%%%%%%%%%%%%%%%%%%%%%%%%50
T & \multirow{2}{*}{35052} &\multirow{2}{*}{1.0\%}&\multirow{2}{*}{50}&\multirow{2}{*}{31.5\%}&  {\bf 25.9}\% &17.3\% &23.2\% &{\bf 25.9}\%\\ 
 U & &&  						 &&1.9\% &0.6\% &0.3\% &{\bf 3.1}\% \\  \midrule
T & \multirow{2}{*}{35052} &\multirow{2}{*}{2.5\%}&\multirow{2}{*}{50}&\multirow{2}{*}{39.6\%}&  32.4\% &22.9\% &28.3\% &{\bf 32.5}\%\\ 
 U & &&  						 &&2.4\% &0.9\% &0.3\% &{\bf 3.9}\% \\   \midrule
T & \multirow{2}{*}{35052} &\multirow{2}{*}{5.0\%}&\multirow{2}{*}{50}&\multirow{2}{*}{47.9\%}&  38.6\% &29.1\% &33.7\% &{\bf 39.6}\%\\ 
 U & &&  						 &&2.4\% &1.3\% &0.4\% &{\bf 5.1}\% \\  \midrule
%%%%%%%%%%%%%%%%%%%%%%%%%%%%%%%%%%%%%%%%%%%%%%%%%%%%%%%%%%%%%%%%%%%100
T & \multirow{2}{*}{17526} &\multirow{2}{*}{1.0\%}&\multirow{2}{*}{100}&\multirow{2}{*}{52.8\%}&  45.2\% &31.8\% &41.3\% &{\bf 46.1}\%\\ 
 U & &&  						 &&2.2\% &0.6\% &0.3\% &{\bf 4.4}\% \\ \midrule
T & \multirow{2}{*}{17526} &\multirow{2}{*}{2.5\%}&\multirow{2}{*}{100}&\multirow{2}{*}{61.3\%}&  52.9\% &38.8\% &47.6\% &{\bf 54.3}\%\\ 
 U & &&  						 &&2.2\% &0.7\% &0.2\% &{\bf 5.1}\% \\  \midrule
T & \multirow{2}{*}{17526} &\multirow{2}{*}{5.0\%}&\multirow{2}{*}{100}&\multirow{2}{*}{68.4\%}&  59.7\% &45.7\% &53.4\% &{\bf 61.3}\%\\ 
 U & &&  						 &&2.2\% &0.8\% &0.2\% &{\bf 5.3}\% \\ \midrule
%%%%%%%%%%%%%%%%%%%%%%%%%%%%%%%%%%%%%%%%%%%%%%%%%%%%%%%%%%%%%%%%%%%%%%250
T & \multirow{2}{*}{6858} &\multirow{2}{*}{1.0\%}&\multirow{2}{*}{250}&\multirow{2}{*}{88.5\%}&  82.1\% &71.2\% &79.3\% &{\bf 85.4}\%\\ 
 U & &&  						 &&1.0\% &0.4\% &0.1\% &{\bf 3.8}\% \\ \midrule
T & \multirow{2}{*}{6858} &\multirow{2}{*}{2.5\%}&\multirow{2}{*}{250}&\multirow{2}{*}{91.8\%}&  86.8\% &77.7\% &83.7\% &{\bf 89.4}\%\\ 
 U & &&  						 &&0.7\% &0.2\% &0.1\% &{\bf 3.0}\% \\  \midrule
T & \multirow{2}{*}{6858} &\multirow{2}{*}{5.0\%}&\multirow{2}{*}{250}&\multirow{2}{*}{93.9\%}&  89.7\% &82.4\% &86.9\% &{\bf 92.0}\%\\ 
 U & &&  						 &&0.5\% &0.2\% &0.0\% &{\bf 2.5}\% \\  \midrule
\end{tabular}
}
\caption{The table contains results of performance tests on the market data. Column {\it Desc} gives the name of performance measure, {\it nr runs} indicates how many subsamples were created for a given $n$, the value $\alpha$ gives the confidence level, and column {\it rejects} gives ratio of subsample for which at least one test rejected the normality assumption. The acronyms JB, AD, SW, and N refer to Jarque--Bera Test, Anderson--Darling Test, Shapiro--Wilk Test, and $N$ test, respectively. Best performance is marked in bold.}
\label{T:tab1}
\end{table}

\begin{figure}[htp!]
\begin{center}
\includegraphics[width=0.45\textwidth]{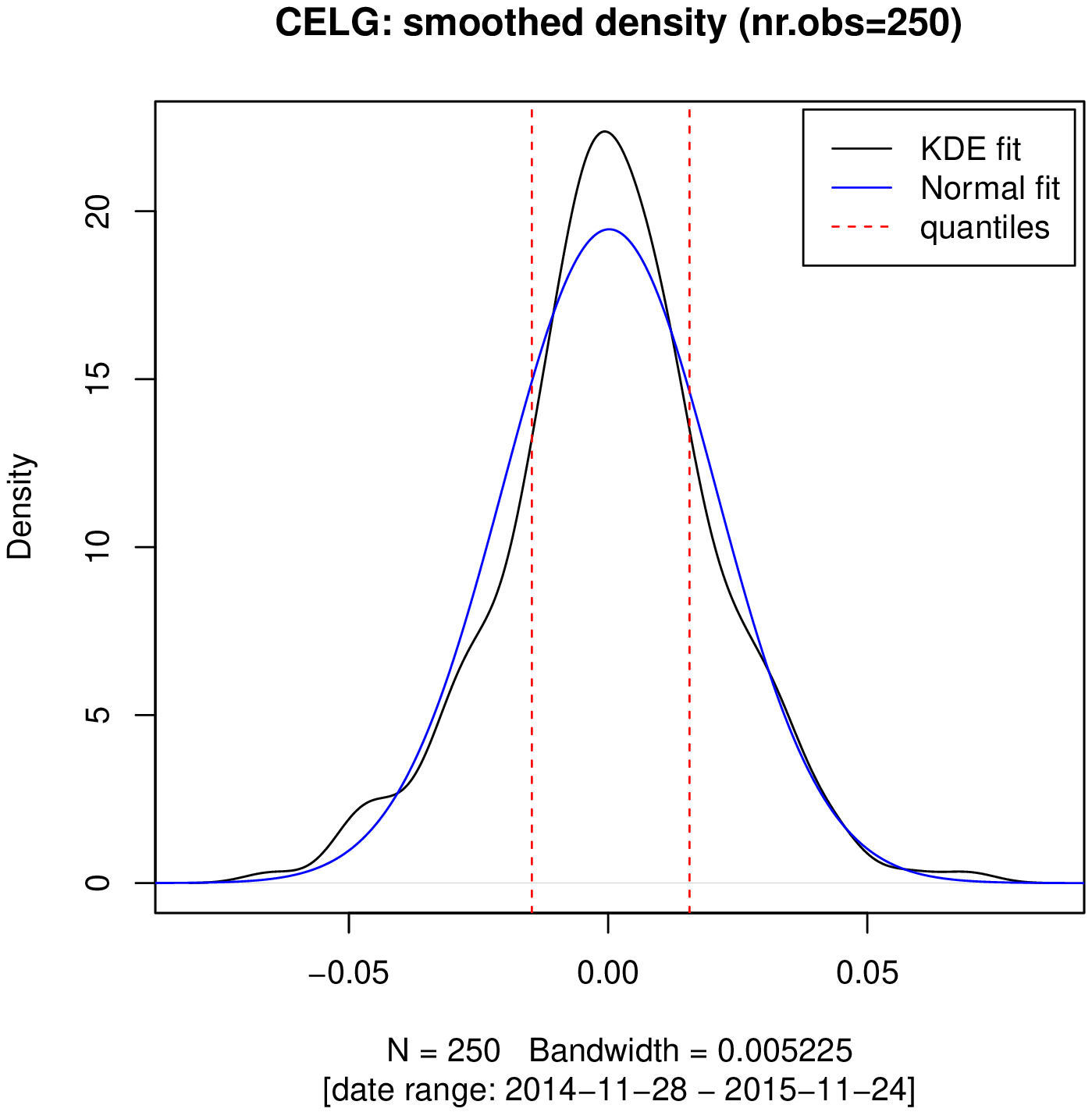}
\includegraphics[width=0.45\textwidth]{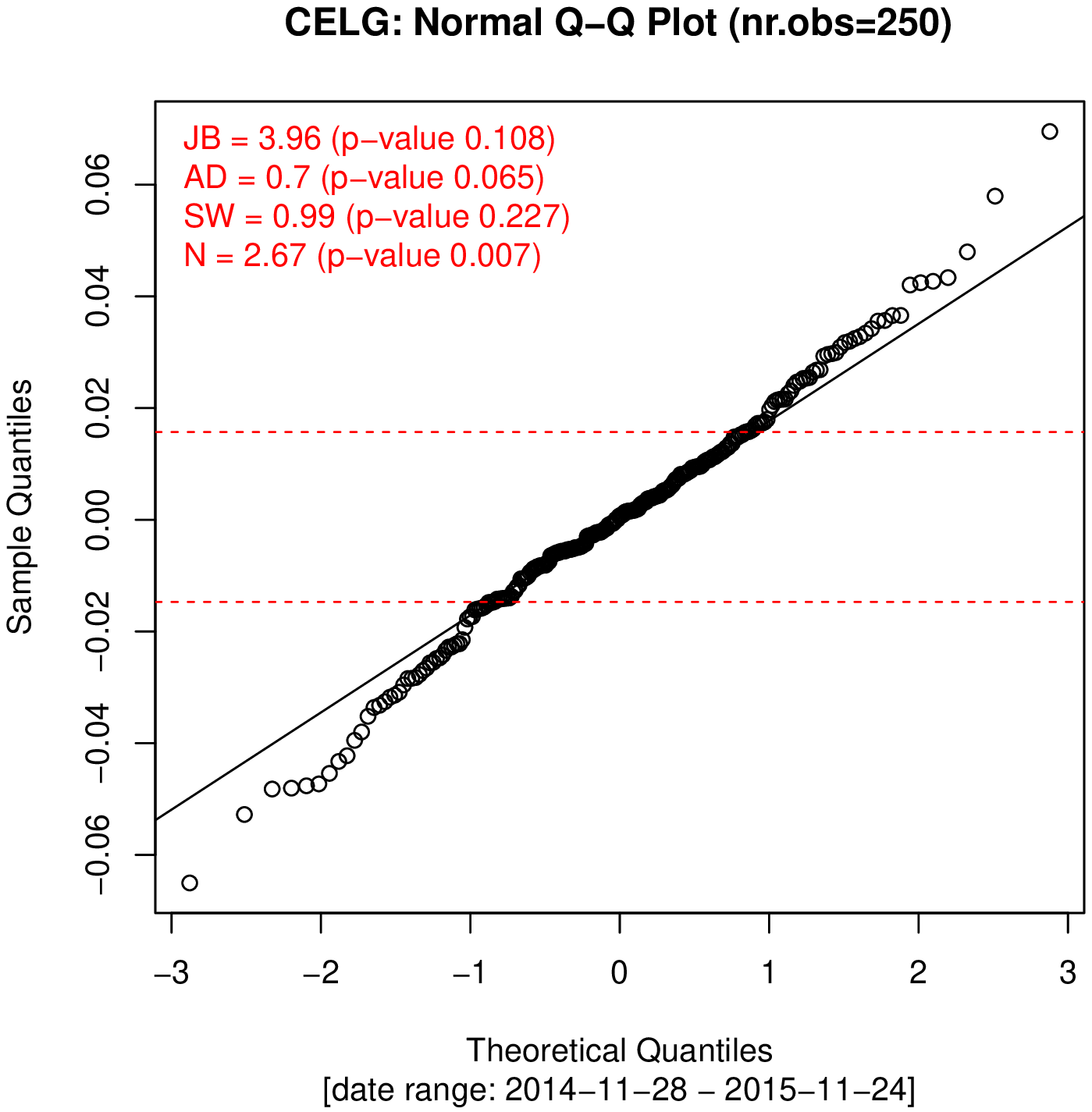}
\end{center}
\vspace*{-0.5cm}
\caption{Exemplary time series for which only test $N$ rejected normality at level 1\%. KDE fit corresponds to empirical density obtained with Kernel Density Estimation, while normal fit correspond to a standard normal fit. Empirical quantiles corresponding to 20/60/20 ratio are marked with dashed line.}
\label{F:fig3}
\vspace*{-0.3cm}
\end{figure}

One can see that the statistic $N$ performs very well and gives the best results for all choices of $n$. Surprisingly, our testing framework allows one to detect non-normal behaviour in cases when other tests fail: the outcomes of measure $U$ are material in all cases. For example, for $n=50$ and $\alpha=5\%$, the value of $U$ was equal to $5.1\%$ -- this corresponds to almost $11\%$ of all rejected samples. The results are especially striking for $n=250$, where the normality assumption was rejected in almost all cases (ca. 90\%). While one might think that for such a big sample size the three classical tests should detect all abnormalities, our test still uniquely rejected normality in multiple cases. For $\alpha=1\%$, the normality was rejected for additional 262 samples ($3.8\%$ of the population). For transparency, in Figure~\ref{F:fig3} we show exemplary data subset for which this happened.

%%%%%%%%%%%%%%%%%%%%%%%%%%%%%%%%%%%%%%%%%%%%%%%%%%%%%%%%%%%%%%%%%%%%%%%%%%%%%%%%%%%%%%%%%%%%%%%%%

\section{Concluding remarks and other applications}\label{S:conculde}

In this paper we have shown that the test statistic $N$ introduced in~\eqref{eq:N} could be used to measure the heaviness of the tails in reference to the central part of distribution and could serve as an efficient goodness-of-fit normality test statistic. Test statistic $N$ is based on the conditional second moments, performs quite well on market financial data, and allows one to detect non-normal behaviour where other benchmark tests fail. 

As mentioned in the introduction, most empirical studies suggest that the normality tests should be chosen carefully as their statistical power varies depending on the context. Our proposal proves to have the best test power in the cases when the true distribution is assumed to be symmetric and have tails that are fatter or slimmer than the normal one. It should be noted that our test is in fact based on the implicit distribution symmetry assumption. Indeed, in \eqref{S:stat.test}, the impact of left and right tail is taken with the same weight. Nevertheless, this could be easily generalised e.g. by considering only one of the tail variances; we comment on that later.

In Theorem \ref{th:asymptotic_distr_test_stat} we proved that the asymptotic distribution of $N$ is normal under the normality null hypothesis. This allows us to study the shape of rejection intervals for sufficiently large samples. To obtain this result, in Lemma \ref{lm:asymptotic_distribution_sigma2M} we derived the asymptotic distribution of the conditional sample variance. 

Also, we showed that the 20-60-20 Rule explains the financial {\it stylised fact} related to tail non-normal behaviour and provides surprisingly accurate clustering of asset return time series. Quite surprisingly, non-normality is visible for almost 40\% of the observations. 

In summary, we believe that tail-impact tests based on the conditional second moments are very promising and provide a nice alternative to the classical framework based e.g. on the third and fourth moments.

For example, the multivariate extension of the test statistic $N$ could be defined using the results presented in \cite{JawPit2015}, e.g. to assess the adequacy of using the correlation structure for dependence modeling. Also, this could be extended to any multivariate elliptic distribution using the results from~\cite{JawPit2017}.

The construction of $N$ shows how to use conditional second moments for statistical purposes. In fact, one might introduce various other statistics that test underlying distributional assumptions. Let us present a couple of examples:

\begin{enumerate}[-]
\item We can test only the (left) low-tail impact on the central part by considering one of test statistics
\[
N_1:= \left(\frac{\hat\sigma^2_L-\hat\sigma^2_M}{\hat\sigma^2}\right)\sqrt{n},\quad
N_2:= \left(\frac{\hat\sigma^2_L-\hat\sigma^2_M}{\hat\sigma_{M}^2}\right)\sqrt{n}.
\]
\item For any quantile-based conditioning sets $A$ and $B$, and any elliptical distribution, one can introduce the statistic
\[
N_3 := \left(\frac{\hat\sigma^2_A}{\hat\sigma^2_B}-\lambda\right)\sqrt{n},
\]
where $\lambda\in\mathbb{R}$ is a constant depending on the quantiles that define conditioning sets and the underlying distribution. Assuming that $A=L$ and $B=\bR$ (whole space), we get the proportion between the tail dispersion and overall dispersion. In this specific case, in the normal framework, we get
\[
\lambda=1-\tfrac{\Phi^{-1}(0.2)\phi(\Phi^{-1}(0.2))}{0.2}-\tfrac{(\phi(\Phi^{-1}(0.2))^2}{0.2^2};
\]
see~\cite[Section 3]{JawPit2015} for details. 
\end{enumerate}
Note that under the normality assumption all proposed statistics are pivotal quantities which facilitates an easy and efficient hypothesis testing; the asymptotic distribution for all statistics could be derived using similar reasoning as the one presented in Theorem~\ref{th:asymptotic_distr_test_stat}.

%%%%%%%%%%%%%%%%%%%%%%%%%%%%%%%%%%%%%%%%%%%%%%%%%%%%%%%%%%%%%%%%%%%%%%%%%%%%%%%%%%%%%%%%%%%%%%%%%

\appendix
\section{Closed-form formula for the normalising constant}\label{A:rho}
In this section, we present the closed-form formula for the normalising constant $\rho$ from Theorem~\ref{th:asymptotic_distr_test_stat}. For brevity, we omit detailed calculations and only present the outcome.

To ease the notation, for any $\gamma \in [0,1]$, we set 
$
x_{\gamma}  :=\Phi^{-1}(\gamma).
$
Then, for any $A=A[\alpha,\beta]$, the standardised second, third, and fourth conditional central moments are given by
\begin{align*}
m_A^{(2)}& :=  1+ \frac{x_{\alpha}\phi(x_{\alpha})-x_{\beta}\phi(x_{\beta})}{\beta-\alpha},\\
m_A^{(3)}& := \frac{(x_{\alpha})^2\phi(x_{\alpha})-(x_{\beta})^2\phi(x_{\beta})}{\beta-\alpha} +2\frac{\phi(x_{\alpha})-\phi(x_{\beta})}{\beta-\alpha},\\
m_A^{(4)}& :=  3+ \frac{(x_{\alpha})^3\phi(x_{\alpha})-(x_{\beta})^3\phi(x_{\beta})}{\beta-\alpha} +3\frac{x_{\alpha}\phi(x_{\alpha})-x_{\beta}\phi(x_{\beta})}{\beta-\alpha}.\footnotemark
\end{align*}
\footnotetext{Recall that for $\alpha=0$ or $\beta=1$ we follow the convention $0\cdot \pm\infty =0.$}
\noindent Moreover, the standardised conditional mean, conditional variance, and conditional kurtosis are equal to
\begin{align*}
\tilde\mu_A & = \frac{\phi(x_{\alpha})-\phi(x_{\beta})}{\beta-\alpha},\\
\tilde\sigma^2_A & = m_A^{(2)} -(\tilde\mu_A)^2,\\
\tilde{\kappa}_A &=\frac{ m_A^{(4)}-3(\tilde\mu_A)^4+6(\tilde\mu_A)^2 m_A^{(2)}-4\tilde\mu_A m_A^{(3)}}{(\tilde{\sigma}^2_A)^2}.
\end{align*}
Also, recall that $\tilde{q}=\Phi(x)$, where $x$ is the unique negative solution of the equation
\[
-x\Phi(x)-\phi(x)(1-2\Phi(x))=0.
\]
Now, we are ready to present the closed-form formula for $\rho$; see Corollary~\ref{cor:rho}.
\begin{corollary}\label{cor:rho}
The normalising constant $\rho$ from Theorem~\ref{th:asymptotic_distr_test_stat} is given by
\begin{align*}
\rho:=&\sqrt{\frac{\tau^2_L}{\sigma^4} +4\frac{\tau^2_M}{\sigma^4} +\frac{\tau^2_R}{\sigma^4} -\frac{4(C_1+C_2)}{\tilde{q}(1-2\tilde{q})}+\frac{2C_3}{\tilde{q}^2}}\,,
\end{align*}
where for $A\in \{L,M,R\}$ we have
\begin{align*}
\frac{\tau^2_A}{\sigma^4} & =  \frac{1}{(\beta-\alpha)^2}\Big((\beta-\alpha)(\tilde\sigma^2_A)^2(\tilde\kappa_A-1) +\alpha(1-\alpha)\left((x_{\alpha}-\tilde\mu_A)^2-\tilde\sigma^2_A\right)^2\nonumber\\
& \quad {}+\beta(1-\beta)\left((x_{\beta}-\tilde\mu_A)^2-\tilde\sigma^2_A\right)^2  -2\alpha(1-\beta)\left((x_{\alpha}-\tilde\mu_A)^2-\tilde\sigma^2_A\right)\times\\
& \quad {}\times\left((x_{\beta}-\tilde\mu_A)^2-\tilde\sigma^2_A\right)\Big),
\end{align*}
and constants $C_1$, $C_2$, and $C_3$ are given by
\begin{align*}
C_1 & = \tilde{q}^2\left((x_{\tilde{q}}-\tilde\mu_L)^2-\tilde\sigma^2_L\right) \left((x_{\tilde{q}}+\tilde\mu_M)^2-\tilde\sigma^2_M\right) \\
&\quad {}-\tilde{q}(1-\tilde{q})\left((x_{\tilde{q}}-\tilde\mu_M)^2  -\tilde\sigma^2_M\right)\left((x_{\tilde{q}}-\tilde\mu_L)^2-\tilde\sigma^2_L\right),\nonumber\\
C_2 & = \tilde{q}^2\left((x_{\tilde{q}}+\tilde\mu_R)^2-\tilde\sigma^2_R\right)\left((x_{\tilde{q}} -\tilde\mu_M)^2-\tilde\sigma^2_M\right) \\
& \quad {}-\tilde{q}(1-\tilde{q})\left((x_{\tilde{q}}+\tilde\mu_M)^2-\tilde\sigma^2_M\right)\left((x_{\tilde{q}}+\tilde\mu_R)^2-\tilde\sigma^2_R\right),\nonumber\\
C_3 & =-\tilde{q}^2\left((x_{\tilde{q}}+\tilde\mu_R)^2-\tilde\sigma^2_R\right)\left((x_{\tilde{q}}-\tilde\mu_L)^2-\tilde\sigma^2_L\right).
\end{align*}
Approximately, the value of $\rho$ is equal to $1.7885$.
\end{corollary}

\section*{Funding}
Part of the work of the second author was supported by the National Science Centre, Poland, via project 2016/23/B/ST1/00479.

 \bibliographystyle{agsm}

 \end{document}